\documentclass[11pt]{article}
\usepackage{amssymb,amsmath,amsthm,amsfonts}
\usepackage{graphicx}
\usepackage{graphicx}
\usepackage{subfigure}
\usepackage[usenames, dvipsnames]{color}
\usepackage{mathrsfs}
\usepackage[CJKbookmarks, colorlinks, bookmarksnumbered=true,pdfstartview=FitH,linkcolor=black]{hyperref}
\usepackage{tcolorbox}
\usepackage{mathdots}
\usepackage{epstopdf}
\graphicspath{{Observer2020127+matlab/}}
\def\disp{\displaystyle}

\def\crr{\cr\noalign{\vskip2mm}}

\def\dref#1{(\ref{#1})}

\theoremstyle{plain}
\newtheorem{theorem}{Theorem}[section]
\newtheorem{lemma}{Lemma}[section]

\newtheorem{assumption}{Assumption}[section]

\numberwithin{equation}{section}

\theoremstyle{definition}
\newtheorem{definition}{Definition}

\newtheorem{remark}{Remark}[section]

\newcommand{\R}{{\mathbb R}}

\def\A{\mathcal{A}}

\begin{document}

\title{{\bf
Extended Dynamics Observer for Linear Systems with Disturbance\footnote{\small
This work is supported by  the National Natural Science
Foundation of China, No.  61873153. }} }

\author{ Hongyinping Feng$^{a}$\footnote{\small  Corresponding author.
Email: fhyp@sxu.edu.cn.}  \  \ and\  \
   \ Bao-Zhu Guo$^{b,c}$  \\
$^a${\it \small School of Mathematical Sciences,}
{\it \small Shanxi University,  Taiyuan, Shanxi, 030006,
China}\\
$^b${\it \small Department of Mathematics and Physics,  }
{\it \small  North
China Electric Power University, Beijing 102206, China}\\
$^c${\it \small Key Laboratory of System and Control, Academy of Mathematics and Systems Science,}
\\{\it \small  Academia Sinica, Beijing, China}
 }
\date{Nov.5,2020}

 \maketitle

\begin{abstract}
This is the last  part of four series papers,
 aiming at stabilization for  signal-input-signal-output (SISO) linear finite-dimensional systems
   corrupted by  general input disturbances.
 A   new observer, referred  to as Extended Dynamics Observer (EDO),
 is proposed  to estimate both the state and disturbance simultaneously.
 The working mechanism of  EDO  consists of two parts: The disturbance with known dynamics is canceled
 completely by its dynamics and the disturbance with unknown dynamics is absorbed by high-gain.
 It is found that the high-gain is always working as long as  the control plant with unknown input disturbance is observable
 which is  the only assumption for the observer design.
 When the disturbance dynamics are completely unknown except some boundedness, the
 EDO is reduced to an extension of the well-known extended state observer or high-gain observer.
 The main advantage of the developed method is that the prior
   information  about both the control plant and the disturbance can be utilized
  as much as possible. The more the prior  information  we have, the better
performance the observer would be.
 An EDO based stabilizing output feedback is also developed in the spirit of
estimation/cancellation strategy.
The stability of the resulting closed-loop system is established  and some of the theoretical results are
 validated by numerical simulations.


\vspace{0.3cm}

\noindent {\bf Keywords:}~ Active disturbance rejection control,    high-gain, internal
model principle, input disturbance,   observer.
\vspace{0.3cm}

\end{abstract}

\section{Introduction}
The dynamic  model of physical systems,  as the prior  information  of control plants, has been
used  in modern control theory as a starting point of the feedback control design.
Since
 the   1960s when  the  control theory was seen as a branch of
applied mathematics,    a fair amount of control strategies such as adaptive control \cite{Adaptive+Control}, optimal control \cite{Optimal2009} as well as nonlinear  control \cite{Nonlinear1992} have been developed  based on mathematical models. These control techniques     make use of   prior  information  about the control plant    as sufficient as possible in the controller design. However, all the model-based feedback laws must be robust to the control plant
uncertainty in engineering applications
so that the ``engineering approximation" can be made \cite{Qianxuesen}.  In other words, the
unknown parts of the control plant, which serve as the ``disturbance'', must be taken into account in the    model-based control design.

The tolerance of disturbance  and uncertainty  is one of the major concerns in modern control theory.
There are many well developed control design approaches to cope with disturbance  in control systems.
 The  adaptive control can be used for the system with unknown parameters  \cite{Adap1} and
  the  robust control is an approach to   achieve robust performance  in the presence of bounded modelling errors \cite{Robust2006}.
 The sliding mode control \cite{Da1} and high-gain control \cite{Highgain+uncertainty1988} work  for
 systems with a large scale of uncertainties.
   The active disturbance rejection control (ADRC)
   has been  recognized as an  almost model free control technology \cite{GuoZHaoBook}.
   Since  it  was  proposed in   the late 1980s  by \cite{Han2009},
    it has been  successfully  applied to  numerous  engineering control problems like
typically control of synchronous motors \cite{ADRCAPP1},  high-speed railway \cite{railway},
 DC-DC power converter \cite{SunGao2005}, flight vehicles control \cite{XiaFu13}, and gasoline engines \cite{Xue2015},
 among many others.
%
%


 As an error driven control technology, ADRC is almost free of mathematical models and
  even works well  for those control plants that are  almost unknown \cite{GaoIJC2017}.
 However, every coin has two sides. On the one hand, the model free characteristic   leads to
  the strong robustness to the uncertainty and  disturbance, and on  the other hand,
  it may waste  more or less some useful  prior information that we  have already  known.
  The waste of the  prior disturbance information also exists    to a varying extent   in
  other control techniques such as
  the
  robust control, high-gain control   and  the sliding mode control.
   In engineering applications, we are not always
   completely ignorant of
     the disturbance. Some
rough information like  smoothness,  boundedness, particularly some dynamic  information
of the disturbance are  available sometimes. This prior  information    might  be useful or even valuable
for the  observer design.  A typical example is    the harmonic disturbance where the known
frequencies are very useful in internal model principle (IMP)  yet are completely wasted in ADRC.
The IMP is an elegant approach  to    robust output regulation, both for finite-dimensional systems \cite{HuangJ2004} and for
 infinite-dimensional ones \cite{Lassi2014SIAM}.  However, the disturbance in  IMP is almost known. Precisely,
  the dynamics of disturbance are required to be  known in IMP, which blocks the general disturbance out the door of the    IMP.
  In one word, a great   improvement room still  exists for both ADRC and IMP  but has not been    noticed and emphasized at least in literature.

In this paper, we develop a  fundamental principle to design  observer via  online measurement information and   prior information about both the  control plant and disturbance. The model of  control plant, as the prior  information  of the system, has been considered sufficiently in
literature. However, the disturbance prior information
 in particular for
 the dynamic modes of  disturbance
is usually    ignored.
 We believe that a good observer
   should possess not only the strong robustness to the disturbance and control plant   but also  the ability to make sufficient  use of all the  valuable prior  information.
 The more the prior  information  is  correctly used, the better performance of the observer would be.
 When the prior  information  is  insufficient, the observer can still do its best.
 In this spirit, a new observer,  referred to as Extend Dynamics Observer (EDO),  is designed
  to estimate both the disturbance and the system state simultaneously.
  The EDO   inherits almost all  the  advantages  from the  extended state observer (ESO) like
 model free characteristic yet can properly utilize the prior  information  not only about the control plant but also
  the disturbance.
 If all the prior dynamic   information  about the total disturbance is available,
the EDO can admit  a zero steady-state error.





Consider  the following  SISO  system:
   \begin{equation} \label{2020151017}
 \left\{\begin{array}{l}
\disp \dot{x} (t) = A  x (t)+B [d(t)+ u(t)] ,\crr
\disp y(t)=Cx (t),
\end{array}\right.
\end{equation}
 where $A \in  \R^{n\times n} $  is the system matrix, $B \in   \R^{n}   $ is the control
 matrix,
    $C\in  \R^{1\times n} $ is the output matrix, $u(t)$ is the control,
  $y(t)$ is the measurement   and $d\in L^2_{\rm loc}[0,+\infty)$ is the  disturbance.
  In this paper, all the unknown signals in the control channel are  referred to as
  disturbances which  may contain  system uncertainties  and  external disturbances.

 If $\hat{d}(t)$ is an estimation of $d(t)$,  a stabilizing  feedback control can be naturally designed as
 \begin{equation} \label{20201261144}
u(t)=-\hat{d}(t)-u_s(t),
\end{equation}
where the first term on the right side is obviously used to
 compensate for the disturbance and the second term $u_s(t)$ is a stabilizer.
  This is referred to as  an estimation/cancellation  strategy and obviously,
  the key point  for  such a strategy
  is the estimation of the  state and disturbance. Different from the
ESO and IMP,  in this work, we decompose  the disturbance into two parts:
the disturbance with known dynamics and the others otherwise. This decomposition is achieved by the mechanism of the system itself automatically.
 The disturbance with known dynamics is treated by likewise observer based  on  IMP
   and the disturbance with unknown dynamics  is dealt with
 by the high-gain which is the core of disturbance estimation in ADRC.
In this way, the prior information can be utilized as sufficient  as possible which remedies the
deficiency of     ADRC and IMP.

The rest of the paper is organized as follows.  In the next section, Section \ref{Dynamics+dis},
we consider the disturbance dynamics and the observability of   system \dref{2020151017}.
   Section \ref{NESO} gives a sufficient condition on  which the  high-gain works.
 Section  \ref{KnownDynamics} is devoted to observer design with   known disturbance dynamics
and Section  \ref{Se.5} is on  observer design for general  disturbance. In Section \ref{ESO},
we focus on systems where the disturbance dynamics is not available at all. A comparison between EDO and ESO is also presented.
Section \ref{Se.7} presents estimation for
general period disturbance which contains harmonic disturbance as a special case.
An observer based output feedback is proposed in Section \ref{SectionFeedback}.
The stability of the closed-loop is also considered. Numerical simulations are presented in  Section
 \ref{Simulation} to validate the theoretical  results, followed up  conclusions in    Section
 \ref{conclusion}.

Throughout the paper, the $n$ and $m$ denote the positive integers and
  the ${\mathbb R}^n$   denotes $n$-dimensional    Euclidean space.
The identity operator  on   $\R^n $  will be denoted by $I_n$ and
the norm of
  $\R^n$  is denoted by $\|\cdot\|_{\R^n}$.
 The spectrum of   operator or matrix  $A$ is denoted by
$\sigma(A)$;
 the largest real part of eigenvalue of $A$
is denoted as $\Lambda_{\rm max }(A)$;
the transpose of matrix $A$ is represented   by $A^\top$.
For simplicity, we denote
$\mathbb{C}_{+}=\{\lambda\in \mathbb{C}\ |\ {\rm Re}\lambda\geq0\}$ and
 $\|\cdot\|_{\infty}=\|\cdot\|_{L^{\infty}[0,\infty)}$.

  \section{   Disturbance dynamics and observability  }\label{Dynamics+dis}

We first consider the disturbance dynamics which  serve   as the   prior information to the disturbance estimation.
Generally speaking,
 not  all continuous disturbances   can be estimated effectively online by a  deterministic  dynamic  system. For instance,
if the disturbance is  a sample path of the Wiener process, it is differentiable for no time $t\geq0$.
 In this case, we do not have any dynamic   information  about the disturbance and
the  estimation  of   such a disturbance  by virtue of typical dynamic system observer seems impossible.
Based on this observation, we first limit ourselves into  an estimable    signal space  of  the following:
\begin{equation}\label{2020251018}
 \mathbb{S}=\left\{s\in L^ {\infty}[0,\infty)\ |\  \dot{s} \mbox{ exists in the weak sense and belongs to } L^{\infty}[0,\infty) \right\} ,
\end{equation}
whose  norm   is given by
\begin{equation}\label{2020251003}
\|s\|_{\mathbb{S}} =  |s(0)|+\| \dot{s} \|_{  \infty},\ \ \forall\ s\in \mathbb{S}.
\end{equation}
A simple computation shows that $(\mathbb{S},\|\cdot\|_{\mathbb{S}} )$ is    a Banach space.
Noting that the
     piecewise    signal   such as
\begin{equation}\label{20201021152}
 s_T(t)=\left\{\begin{array}{ll}
 e^t,&t\in[0,T],\\
 e^T,&t\geq T
 \end{array}\right.
 \end{equation}
belongs to $\mathbb{S}$,
the signal space $\mathbb{S}$ is quite general and can include the harmonic signals, bounded   continuously differentiable   periodic signals,
    piecewise   polynomial signals,
piecewise exponential signals and their linear combinations.

  Let $(G,Q)$ be an  observable system
 with the state space $\R^m$ and output space $\R$.
 Define
\begin{equation}\label{201912311058}
\begin{array}{l}
\disp \Omega   ( G)= \Big{\{} Qv(t)  \ \Big{|}\
 \dot{v}(t)=Gv(t),\ v(0)\in {\R^{m}},\
t\in   \R \Big{\}}.
\end{array}
\end{equation}
 By ordinary differential equation
theory, we obtain
\begin{equation}\label{20195181211}
\begin{array}{l}
\disp \Omega   ( G)={\rm span}\Big{\{} t^{m_{\lambda}-k}e^{ \lambda t}   \ \Big{|}\
\lambda  \in \sigma(G),   k=1,2,\cdots,m_{\lambda},\crr
\hspace{2cm}\disp \ \ m_{\lambda} {\rm \ is \ the \ algebraic\ multiplicity\ of\  }\lambda,\
t\in   \R \Big{\}},
\end{array}
\end{equation}
  which implies that the space $\Omega (G)$  is independent of $Q$.
By  \dref{201912311058},
$\Omega (G)\subset \mathbb{S}$ as long as  $\sigma(G)\subset i\R$ and each eigenvalue of $G$ is algebraically simple.  Define the  projection   operator  $\mathbb{P}_G:  \mathbb{S }\to \Omega(G)$ by
\begin{equation}\label{20201211811}
 \mathbb{P}_G s =\arg\inf_{g\in \Omega(G) } \| s-g\|_{ \mathbb{S }},\ \ \forall\ s(\cdot)\in  \mathbb{S }.
\end{equation}
Since $(\mathbb{S},\|\cdot\|_{\mathbb{S}} )$ is    a Banach space,
the optimal approximation $\mathbb{P}_G s\in\Omega(G)$ always exists, which implies that the operator  $\mathbb{P}_G$  is well defined.
%
%
%
Let $e =(I-\mathbb{P}_\mathbb{S})s$ be the approximation error. A simple computation shows that $e(0)=0$ and thus
\begin{equation}\label{2020251843}
\|e\|_{\mathbb{S}}=\|\dot{e}\|_{\infty}.
\end{equation}
In fact, if $s_*(\cdot)=\mathbb{P}_\mathbb{G} s$  with $s_*(0)\neq s(0)$, i.e., $e(0)\neq0$,  then
\begin{equation}\label{2020251919}
\|s_*-s_*(0)+s(0)-s\|_{\mathbb{S}}=   \|\dot{s}-\dot{s}_*\|_{\infty}
< \|s_*-s\|_{\mathbb{S}}.
\end{equation}
Since $s_*(\cdot)-s_*(0)+s(0)\neq s^*(\cdot)$, \dref{2020251919}  contradicts to the optimality of  $s_*(\cdot)=\mathbb{P}_\mathbb{S} s$ for  $s(\cdot)$.

\begin{definition}\label{De202010021011}

Let $A \in \R^{n\times n}$, $B \in  \R^n  $
 and
    $C\in  \R^{1\times n}$.
    Suppose that $\Theta$ is a set of signals  and we have known that  $d\in \Theta$.
 System \dref{2020151017}  is said to be observable for the signal set   $\Theta$,
  provided both the initial state and the disturbance are    distinguishable in the sense that:
  For any $T>0$,
   \begin{equation} \label{20201021019}
 u(t)=0\ \mbox{and}\  y(t)= 0 \;  \hbox{ for a.e.} \  t\in[0,T]\ \Rightarrow\ x(0)=0  \mbox{ and }  d(t)=0
 \; \hbox{ for a.e.} \ t\in[0,T].
  \end{equation}
\end{definition}

\begin{lemma}\label{Lm20201021046}
Suppose that system \dref{2020151017} takes on the
 observability canonical form, i.e.,
 \begin{equation}\label{202010021047}
		A =\begin{bmatrix}
0&0& \cdots&0& a_1\\
1&0& \cdots&0& a_{2}\\
 0&1& \cdots&0& a_{3}\\
\vdots&\vdots&\ddots&\vdots&\vdots \\
		0&0& \cdots&1& a_n\\
		\end{bmatrix},  \ \ B=
\begin{bmatrix}
 b_1\\b_2\\
 \vdots \\b_n\\
		\end{bmatrix}\neq0\ \  \mbox{and}\ \ C=[0\ 0\ \cdots\ 0\ 1]\in \R^{1\times n},
		\end{equation}
where $a_j ,b_j \in\R$, $j=1,2,\cdots,n$.  Then,
system \dref{2020151017}  is  observable for    $\mathbb{S}$ if and only if
  $b_2=b_3=\cdots=b_n=0$ and $b_1\neq0$.

\end{lemma}
\begin{proof}
 Suppose that $b_1\neq0$ and $b_2=b_3=\cdots=b_n=0$. Then, for any $T>0$,  $u(t)=0$
   for a.e.  $t\in [0,T]$ implies that
  \begin{equation}\label{202010021055}
\left\{\begin{array}{l}
\dot{x}_1(t)=a_1x_n(t)+b_1d(t), \crr
 \dot{x}_2(t)=x_1(t)+a_2x_n(t),\crr
 \cdots\cdots\cdots\cdots\cdots\cdots\cdots\crr
 \dot{x}_n(t)=x_{n-1}(t)+a_nx_n(t),
\end{array}\right. \hbox{ a.e.}\   t\in[0,T],
 \end{equation}
 where $x(t)=[x_1(t)\  \ x_2(t)\ \ \cdots\ \ x_n(t)]^{\top}$.
If  $  y(t)= x_n(t)=0$  for a.e.  $t\in [0,T]$,
   \dref{202010021055} yields $$
 x_n(t)=x_{n-1}(t)=\cdots=x_{ 1}(t)=0,\ \ \hbox{ a.e.}\ \    t\in [0,T],
 $$
 which implies that $d(t)=0$ for a.e. $t\in [0,T] $ due to $b_1\neq0$.
 Hence,  system \dref{2020151017}  is  observable for    $\mathbb{S}$.

 Conversely, suppose that system \dref{2020151017}  is  observable for    $\mathbb{S}$.
We first claim that $b_n=0$.  Otherwise, for any $T>0$,
\begin{equation}\label{202010021104}
\left\{\begin{array}{l}
\dot{x}_1(t)=a_1x_n(t)+b_1d(t), \crr
 \dot{x}_2(t)=x_1(t)+a_2x_n(t)+b_2d(t),\crr
 \cdots\cdots\cdots\cdots\cdots\cdots\cdots\crr
 \dot{x}_n(t)=x_{n-1}(t)+a_nx_n(t)+b_nd(t),\crr
  y(t)= x_n(t)=0,
\end{array}\right. \hbox{ a.e.}\ t\in[0,T]
 \end{equation}
 implies that $x_{n-1}(t)+b_nd(t)=0$ for a.e. $t\in[0,T]$ and hence system \dref{202010021104} turns out to be
 \begin{equation}\label{202010021106}
\left\{\begin{array}{l}
\disp \dot{x}_1(t)= -\frac{b_1}{b_n}x_{n-1}(t), \crr
\disp  \dot{x}_2(t)=x_1(t) -\frac{b_2}{b_n}x_{n-1}(t),\crr
 \cdots\cdots\cdots\cdots\cdots\cdots\cdots\crr
 \disp \dot{x}_{n-1}(t)=x_{n-2}(t) -\frac{b_{n-1}}{b_n} x_{n-1} (t),
\end{array}\right. \hbox{ a.e.}\ t\in[0,T].
 \end{equation}
Since  system \dref{202010021106} with the  output $x_{n-1}(\cdot)$ is
  of the observability  canonical form,  it  is
 always observable for any
 $b_j\in\R$, $j=1,2,\cdots,n$. As a result, each non-zero solution of  system \dref{202010021106}
 satisfies $x_{n-1}(t)=-b_nd(t)\neq0$ for a.e. $t\in[0,T]$ and hence
 is the
 zero dynamics of the original system \dref{202010021104}. This contradicts to the observability of
 system \dref{2020151017}. We hence obtain  $b_n=0$.
 Similarly, we can prove that $b_{n-1}=0$. Indeed, in this case,  for any $T>0$,
 \begin{equation}\label{202010021135}
\left\{\begin{array}{l}
\dot{x}_1(t)=a_1x_n(t)+b_1d(t), \crr
 \dot{x}_2(t)=x_1(t)+a_2x_n(t)+b_2d(t),\crr
 \cdots\cdots\cdots\cdots\cdots\cdots\cdots\crr
 \dot{x}_{n-1}(t)=x_{n-2}(t) +b_{n-1}d(t),\crr
  y(t)= x_{n-1}(t)=0,
\end{array}\right.\ \ \hbox{ a.e.}\ t\in[0,T]
 \end{equation}
  implies that $x_{n-2}(t)+b_{n-1}d(t)=0$  for a.e. $ t\in[0,T] $  and hence system \dref{202010021135} is reduced to
 \begin{equation}\label{202010021138}
\left\{\begin{array}{l}
\disp \dot{x}_1(t)= -\frac{b_1}{b_{n-1}}x_{n-2}(t), \crr
\disp  \dot{x}_2(t)=x_1(t) -\frac{b_2}{b_{n-1}}x_{n-2}(t),\crr
 \cdots\cdots\cdots\cdots\cdots\cdots\cdots\crr
 \disp \dot{x}_{n-2}(t)=x_{n-3}(t) -\frac{b_{n-2}}{b_{n-1}} x_{n-2} (t),
\end{array}\right.   \hbox{ a.e.}\ t\in[0,T].
 \end{equation}
 Since  system \dref{202010021138} with the output $x_{n-2}(\cdot)$ is
   always observable for any
 $b_j\in\R$, $j=1,2,\cdots,n-1$, each non-zero solution of  system \dref{202010021138} is a
 zero dynamics of the original system \dref{202010021104}. This contradicts to the observability of
 system \dref{2020151017}. We hence obtain  $b_n=b_{n-1}=0$.
 Moreover, we can obtain $b_n=b_{n-1}=b_2=0$ by repeating the same process. This completes the proof
 of the lemma due to  $B\neq0$.
\end{proof}

%
%
%
%

 \begin{lemma}\label{lm20209132036}
  Let $A\in \R^{n\times n}$ and $G\in\R^{m\times m}$. Suppose that
 \begin{equation} \label{20209132047}
 \sigma(A)\cap\sigma(G)=\emptyset.
   \end{equation}
 Then, system \dref{2020151017}  is observable for  $\Omega(G)$ if and only if   $(A,C)$ is observable and  the following  transmission zeros condition holds:
   \begin{equation} \label{2020213905}
C(\lambda-A)^{-1}B\neq 0, \ \ \ \forall\ \lambda\in \sigma(G).
\end{equation}

 \end{lemma}
\begin{proof}
Since we have known that  $d\in \Omega$, there exists a $Q\in \R^{1\times m}$ such that
$(G,Q)$ is observable and the disturbance can be written as  $\dot{v}(t)=Gv(t)$ and $d(t)=Qv(t)$ for some initial state. As a result, system \dref{2020151017}  takes the form
 \begin{equation} \label{202010021732}
 \left\{\begin{array}{l}
\disp \dot{x} (t) = A  x (t)+B  [Q v(t)+ u(t)] ,\crr
\disp \dot{ v} (t) =  G v  (t),  \ \
\disp y(t)=Cx (t).
\end{array}\right.
\end{equation}
 If we define
 \begin{equation} \label{20209132035}
  \A_e=\begin{bmatrix}
          A&BQ \\
          0&G
        \end{bmatrix}\ \ \mbox{and}\ \ \mathcal{C}_e= [C\quad 0],
   \end{equation}
then  system \dref{2020151017}  is observable for  $\Omega(G)$ if and only if
     system  $(\A_e,\mathcal{C}_e)$ is  observable.

  Suppose that  $(\mathcal{A}_e,\mathcal{C}_e)$ is observable and there exists
 a $\tau>0$ such that $Ce^{At}x\equiv0$ for any $t\in[0,\tau]$. Then,
 $\mathcal{C}_ee^{\A_e t}(x,0)^{\top}=Ce^{At}x\equiv0$ implies that
   $x=0$ and hence   $(A,C)$ is observable.
  For any $\lambda\in \sigma(G)\subset\sigma(\A_e)$, suppose that
$\A_e(x,v)^{\top}=\lambda(x,v)^{\top} $ with $(x,v)^{\top}\neq0$. By exploiting   \cite[p. 15, Remark 1.5.2]{Weissbook} and the observability of  $(\mathcal{A}_e,\mathcal{C}_e)$, a simple computation shows that
 \begin{equation}\label{20209132157}
 \mathcal{C}_e  (x,v)^{\top}=  C x =C(\lambda-A)^{-1}BQv\neq 0,
 \end{equation}
which leads to   \dref{2020213905} easily.

Conversely, for any $\lambda\in \sigma(G)\subset\sigma(\A_e)$, suppose that
$\A_e(x,v)^{\top}=\lambda(x,v)^{\top} $ and $\mathcal{C}_e  (x,v)^{\top}=0$.
Then, $C(\lambda-A)^{-1}BQv= 0$ and $Gv=\lambda v$. By   assumption \dref{2020213905},   $Qv=0$ and hence
$v=0$ by  the observability of $(G,Q)$.
As a result,  the equations $\A_e(x,v)^{\top}=\lambda(x,v)^{\top} $ and $\mathcal{C}_e  (x,v)^{\top}=0$ are reduced to  $Ax=\lambda x$  and $C x = 0 $.
By the observability of $(A,C)$, we obtain   $x=0$.
Therefore, $(\mathcal{A}_e,\mathcal{C}_e)$  is observable, or equivalently,
system \dref{2020151017}  is observable for  $\Omega(G)$.
\end{proof}

  We   point out that the
observability of disturbance corrupted system \dref{2020151017} depends on the disturbance set $\Theta$ which   serves as
the prior disturbance information  we have known.
 Different disturbance set  may lead to different observability even for the same system.
Here is an  example to show this point. Let
 \begin{equation} \label{20209171839}
  A =\begin{bmatrix}
0&1 \\
		0&0
		\end{bmatrix},\ \
 B=\begin{bmatrix}
 		0 \\
		  1
		\end{bmatrix}\ \ \mbox{and} \ \ C=[1\ \ -1 ] .
 \end{equation}
 Then, system \dref{2020151017} with $u=0$  can be  written as
   \begin{equation}\label{20209171849}
\left.\begin{array}{l}
\dot{x}_1(t)=x_2(t), \ \dot{x}_2(t)=d(t),\ \ y(t)=x_1(t)-x_2(t).
\end{array}\right.
 \end{equation}
 Suppose that we know  nothing about the disturbance except  $d\in \mathbb{S}$. Then, system
\dref{20209171849} is not observable for  $ \mathbb{S}$. Indeed, a simple computation shows that
$ x_1(t)  =x_2(t)=d(t)=s_T(t)$
   is a nonzero solution of system \dref{20209171849} over $[0,T]$, where $s_T$ is given by \dref{20201021152}.  However, $d\in \mathbb{S}$
   and
the  output  satisfies $y(t)=x_1(t)-x_2(t)\equiv0$ on $[0,T]$. By Definition \ref{De202010021011},
system
\dref{20209171849} is not observable for  $ \mathbb{S}$.
If we have known the dynamics of the disturbance, the situation becomes completely different.
Suppose that we have known $d\in \Omega(G)$ for some matrix $G$ satisfying $0,1\notin\sigma(G)$. Then,
there exists a vector  $Q$  such that system $(G,Q)$  is observable and hence
  system  \dref{20209171849} can be written as
 \begin{equation}\label{20209171849G}
\left\{\begin{array}{l}
\dot{x}_1(t)=x_2(t), \ \dot{x}_2(t)=Qv(t),\crr
\dot{v}(t)=Gv(t),\crr
y(t)=x_1(t)-x_2(t),
\end{array}\right.
 \end{equation}
 which is a disturbance free system.
   By Lemma \ref{lm20209132036}, it is easy to see that system \dref{20209171849G} is observable. In other words, system \dref{2020151017}
  is observable for  $\Omega(G)$ which  is completely different from the observability for    $\mathbb{S}$.
   This fact implies that, if the prior information about the disturbance  is enough,  we may still  estimate the disturbance $d(\cdot)$ from
  system \dref{2020151017} in terms of   the output $y(\cdot)$  even if it is unobservable for $\mathbb{S}$.

\begin{remark}\label{Re202010031150}
   Definition \ref{De202010021011}
 is different from the observability
of  disturbance free system where the observability on some  finite interval  $[0,T]$
implies   the observability on entire  $[0,\infty)$.
Owing to the uncertainty of disturbance,
 it is almost  impossible to estimate the disturbance on  $[T,\infty) $ by  the information of output over   $[0,T]$.
 \end{remark}


\section{High-gain for stabilization}\label{NESO}


In most of the cases, we have to  pay prices in estimating disturbance from measured output and the prices are usually characterized by the high-gain.
 Since it  does not need necessarily  the prior   information   about  disturbance
except for  some rough information like  boundedness,
the  high-gain is  an effective and practical way to cope with the disturbance. In  \cite{KhalilTAC2008}, it has been used to the observer design for the   system that represents a chain of $n$ integrators.
 The well-known ESO in ADRC is also by  means of the high-gain \cite{GAO2003}, \cite{Han2009}.
 In this section, we will consider the basic principle of  high-gain and
 investigate  the relationship between the observability
 and the    high-gain.


To show  the   basic principle of high-gain clearly, we begin with the direct  propositional feedback for   a scalar system with input disturbance:
   \begin{equation}\label{20209171831}
\dot{x}(t)=u(t)+d(t), \ \ u(t)=-\omega x(t),
 \end{equation}
where    $d\in L^{\infty}[0,\infty)$ is the disturbance and
 $\omega$ is a positive tuning parameter. We solve the closed-loop straightforwardly to  get
 \begin{equation}\label{20209171949}
 |{x}(t)|\leq e^{-\omega t}|x(0)|+\int_0^te^{-\omega(t-s)}|d(s)|ds\leq e^{-\omega t}|x(0)|+\frac{\|d\|_{\infty}}{\omega}.
 \end{equation}
 That is
  \begin{equation}\label{20209302046}
 \lim_{t\to\infty}|{x}(t)|\leq  \frac{\|d\|_{\infty}}{\omega},
 \end{equation}
 which implies that we can stabilize $x(\cdot)$ as small as possible  by increasing the feedback gain $\omega$.
In other words,  the negative impact of the  disturbance $d(\cdot)$  in  system \dref{20209171831} can be
 eliminated    by increasing the feedback gain  $\omega$.
 However, this property seems  not trivial  for general linear systems. Here is a sufficient condition under which the high-gain works.

%
%
%



\begin{lemma}\label{Lm20209172011}
 Let  $A_{\omega} \in \R^{n\times n}$ be  a Hurwitz matrix  with
 $\omega=-\Lambda_{\max}(A_{\omega})>0$.  Suppose that $B\in \R^n$ such that
   \begin{equation}\label{20209172022}
  \lim_{\omega\to+\infty}\|(s-A_{\omega})^{-1}B\|_{\R^n}=0  \mbox{ uniformly on }  s\in  \mathbb{C}_+.
  \end{equation}
 Then, there exists an $L_B>0$, independent of $\omega$ and $t$, such that
  \begin{equation}\label{20209191648}
 \|e^{A_{\omega}t} B\|_{\R^n}\leq L_Be^{-\omega t},\ \ t\geq0.
  \end{equation}
 As a result,   for any
    $d\in L^{\infty}[0,\infty)$,  the solution of system
     $    \dot{x}(t)=A_\omega x(t)+Bd(t)$
     satisfies
   \begin{equation}\label{20201001959}
 \lim_{t\to\infty}\|{x}(t)\|_{\R^n}\leq  \frac{L_B\|d\|_{\infty}}{\omega}.
 \end{equation}

    \end{lemma}
    \begin{proof}
  Let $ \varepsilon_j= [0\ \cdots\  0\   1_{j{\rm th}}\ 0\ \cdots\ 0]^{\top}$ denote the
$j$-th coordinate vector where $1_{j{\rm th}}$ means the component in the $j$-th position is 1$,
j =1,2,\cdots,n$.
By the assumption  \dref{20209172022},
  \begin{equation}\label{20209172055}
\lim_{\omega\to+\infty}|\varepsilon_j^{\top}(s-A_{\omega})^{-1}B|=0  \mbox{ uniformly on }  s\in  \mathbb{C}_+
\end{equation}
for   $j =1,2,\cdots,n$.
Applying the inverse Laplace transform to \dref{20209172055}, we obtain
\begin{equation}\label{20209232121}
\begin{array}{l}
\disp \lim_{\omega\to\infty}|\varepsilon_j^{\top}e^{A_{\omega}t}B|
  \disp =
 \frac{1}{2\pi i}\lim_{\omega\to+\infty}\lim_{T\to\infty}\int_{\gamma-iT}^{\gamma+iT}
 e^{st}\varepsilon_j^{\top}(s-A_{\omega})^{-1}Bds\crr
   \hspace{1cm}\disp =
 \frac{1}{2\pi i}\lim_{T\to\infty}\int_{\gamma-iT}^{\gamma+iT}
 e^{st}\lim_{\omega\to+\infty} \varepsilon_j^{\top}(s-A_{\omega})^{-1} Bds\crr
  \hspace{1cm} \disp =
 \frac{1}{2\pi i}\lim_{T\to\infty}\int_{\gamma-iT}^{\gamma+iT}
 e^{st}0ds=0
  ,\ \  t\geq0,
  \end{array}
\end{equation}
where $\gamma$ is a real number so that the contour path of the  integration is in the region of convergence of $\varepsilon_j^{\top}e^{A_{\omega}t}B$,  $j =1,2,\cdots,n$.
 Since $A_{\omega}$ is Hurwitz,  \dref{20209191648}   follows  from \dref{20209232121} easily.
 Moreover,  \dref{20201001959} holds due to
 $$
 x(t)=e^{A_{\omega}t}x(0)+\int_0^te^{A_{\omega} s}Bd(t-s)ds.
 $$
    \end{proof}

\begin{remark}\label{Re2020117}
We point out that \dref{20209191648} does not hold for all controllable systems.
For example, if we choose
  \begin{equation} \label{202091171005}
  A_{\omega} =\begin{bmatrix}
0&1 \\
		-\omega^2&-2\omega
		\end{bmatrix}, \  \   \ \
 B=\begin{bmatrix}
 		 1 \\
		  1
		\end{bmatrix},\ \ \omega>0 ,
 \end{equation}
 then $(A_\omega,B)$ is controllable. However, a straightforward computation shows that
 \begin{equation} \label{20201171012}
e^{A_{\omega}t}B=e^{-\omega t}\begin{bmatrix}
\disp 1+(\omega    +1) t  \crr
\disp 1-(\omega^2+\omega)t
 \end{bmatrix},\ \ t\geq0
\end{equation}
and in particular,
 \begin{equation} \label{20201171027}
\left\|e^{A_{\omega}\frac{1}{\omega}}B\right\|_{\R^2}=e^{-1}\left\|\begin{bmatrix}
\disp 2    +1/\omega     \crr
\disp -\omega
 \end{bmatrix}\right\|_{\R^2} \to \infty\ \ \mbox{as}\ \ \omega\to+\infty.
\end{equation}
 \end{remark}

 The following Theorem shows that system \dref{2020151017}   can always be stabilized to zero  by high-gain
    provided it is observable for  $ \mathbb{S}$.

\begin{theorem}\label{Th202010021237}
    Suppose that  system \dref{2020151017} is observable for  $\mathbb{S}$. Then, system $(A,B)$ is controllable and there exist functions $f_j\in C[0,\infty)$, $j=1,2,\cdots,n$
  such that the feedback
     \begin{equation} \label{202010021359}
 u(t)=[ f_1(\omega)\ f_2(\omega)\ \cdots\ f_n(\omega)]^{\top} x(t),\ \ \omega>0,
 \end{equation}
stabilizes  system
\begin{equation} \label{202010021521}
\dot{x}(t)=Ax(t)+B[d(t)+u(t)],\ \ d\in \mathbb{S}.
 \end{equation}
 In other words,    the closed-loop system given by
\dref{202010021521} and \dref{202010021359} satisfies:
 \begin{equation} \label{202010021300}
 \lim_{t\to\infty}\|x(t)\|_{\R^n}\leq \frac{M\|d\|_{\infty}}{\omega},\ \
 \end{equation}
 where $M$ is a positive constant that is independent of $\omega$.
   \end{theorem}
    \begin{proof}
       We assume without loss of the generality that $A$, $B$ and $C$ are given  by    the observability canonical form    \dref{202010021047}. By
      Lemma \ref{Lm20201021046}, we conclude that $b_1\neq0$ and $b_2=b_3=\cdots=b_n=0$.  For simplicity, we suppose that $b_1=1$. Since system $(A^{\top},B^{\top})$ is a chain of $n$ integrators, it is observable.
      So $(A,B)$ is also controllable. As a result, there exists an
        invertible transformation $U$    that converts
        system $(A,B)$  into  the controllability canonical form $(A^{\top}, C^{\top})$.
        More specifically,
   \begin{equation}\label{2019141433}
\left.\begin{array}{l}
\disp   UAU^{-1}=A^{\top} \ \ \mbox{and}\ \ UB= C^{\top}.
\end{array}\right.
\end{equation}
It is sufficient to consider   the following system:
 \begin{equation}\label{2019141428}
\left.\begin{array}{l}
\disp  \dot{z}(t)=A^{\top}z(t)+C^{\top}[d(t)+u(t)].
\end{array}\right.
\end{equation}
Let
  \begin{equation} \label{2020927830}
  K_{\omega }=[
 k_1\omega ^n -a_1\quad
  k_2\omega ^{n-1} -a_2\quad
  \cdots \quad
  k_n\omega -a_n] ,\ \ \omega>0,
 \end{equation}
 where
 $K=[k_1-a_1\quad  k_2-a_2\quad \cdots\quad k_n-a_n] $ is a vector such that
 $ A ^{\top}+C^{\top}K $
is  Hurwitz. A simple computation shows that
 \begin{equation}\label{202010021517}
		A ^{\top}+C^{\top}K_{\omega} =		 \begin{bmatrix}
0&1&0&\cdots&0\\
		0&0&1&\cdots&0\\
		\vdots&\vdots&\vdots&\ddots&\vdots\\
		0&0&0&\cdots&1\\
		  k_1\omega^n& k_2\omega^{n-1}&  k_3\omega^{n-2}&\cdots&  k_n\omega
		\end{bmatrix}
 \end{equation}
  is Hurwitz as well and
    \begin{equation}\label{202010021524}
 \lambda \omega \in \sigma( A ^{\top}+C^{\top}K_{\omega})\ \ \mbox{if and only if}\ \ \lambda  \in \sigma(A ^{\top}+C^{\top}K) .
 \end{equation}
 Moreover, for any $s\in \mathbb{C}_{+}$, it follows that
 \begin{equation}\label{02010021525}
		\left[s-( A ^{\top}+C^{\top}K_{\omega})\right]^{-1}C^{\top}  =\frac{-1}{k_1 \omega^n + k_2 \omega^{n-1} s + \cdots+ k_n\omega s^{n-1} - s^n} \begin{bmatrix}
  1 \\
    s \\
    \vdots\\
  s^{n-1}
\end{bmatrix}.
\end{equation}
Since  $A ^{\top}+C^{\top}K_{\omega}$
is Hurwitz, we have $k_1\neq0$ and hence
 \begin{equation}\label{02010021526}
  \lim_{\omega\to+\infty}\|[s-( A ^{\top}+C^{\top}K_{\omega})]^{-1}C^{\top}\|_{\R^n}=0\ \ \mbox{uniformly on}\  s\in \mathbb{C}_{+}.
  \end{equation}
   By Lemma \ref{Lm20209172011},  there exists an $L_C$ that is independent of $\omega$ and $t$ such that
       the solution of system
      $    \dot{z}(t)=(A ^{\top}+C^{\top}K_{\omega}) z(t)+C^{\top}d(t)$
     satisfies
   \begin{equation}\label{202010021539}
 \lim_{t\to\infty}\|{z}(t)\|_{\R^n}\leq  \frac{L_C\|d\|_{\infty}}{\omega}.
 \end{equation}
 By \dref{2019141433}, we can obtain \dref{202010021300} easily and moreover, the feedback \dref{202010021359} is given by
  \begin{equation} \label{202010031029}
 u(t)= K_{\omega}U^{-1}x(t),\ \ \omega>0.
 \end{equation}
    \end{proof}


\section{Observer design with   known disturbance dynamics}\label{KnownDynamics}
In this section,  we consider a special case that the disturbance dynamics are known, i.e.,  $d\in \Omega(G)$ with known $G\in\R^{m\times m}$. Since the prior information about the disturbance is sufficient, this is  the simplest case for the observer design yet is the concise situation to demonstrate the new
idea of observer design.

Suppose that $(G,Q)$ is   observable with output space $\R$.
 Consider the following
  Luenberger observer of system \dref{2020151017}:
\begin{equation} \label{20201021624}
 \left\{\begin{array}{l}
\disp \dot{\hat{x}} (t) =  A \hat{x} (t)+B Q \hat{x} (t)-F_1[y(t)- C \hat{x} (t)]+B u(t),\crr
\disp  \dot{\hat{v}} (t) =G\hat{v} (t)+F_2[y(t)-C\hat{x} (t)],
\end{array}\right.
\end{equation}
where
 $F_1\in   \R^{n}    $  and $F_2\in  \R^{m}   $  are the gain vectors  to be determined. When system  \dref{20201021624} is observable, $F_1$ and $F_2$ can be chosen easily by
the pole assignment theorem. However, we will choose $F_1$ and $F_2$ in another way so that we can cope with the general disturbance by high-gain in Section \ref{Se.5}.
 Let the  observer errors be
\begin{equation} \label{wxh20202061802Ad102}
\tilde{x} (t)=x (t)-\hat{x} (t) \ \ \mbox{and}\ \ \tilde{v}(t)=v(t)-\hat{v}(t).
\end{equation}
Then, they are  governed by
\begin{equation} \label{wxh20202061803Ad102}
 \left\{\begin{array}{l}
\disp \dot{\tilde{x}} (t) = (A +F_1 C )\tilde{x} (t)+B Q \tilde{x} (t),\crr
\disp  \dot{\tilde{v}} (t) =  G \tilde{v} (t)-F_2C \tilde{x} (t).
\end{array}\right.
\end{equation}
If we  select  $F_1$ and $F_2$ properly such that  system \dref{wxh20202061803Ad102} is stable, then
$(x (t),v(t))$  can be estimated in the sense that
\begin{equation} \label{wxh20202062106Ad102}
  \|(x (t)-\hat{x} (t) ,v(t)-\hat{v} (t))\|_{\R^n \times \R^m}\to  0
  \ \ \mbox{as}\ \ t\to\infty.
\end{equation}
Inspired by the first two parts \cite{FPart1} and \cite{FPart2}  of this series works, the $F_1$ and $F_2$ can be chosen
easily by decoupling  the   system \dref{wxh20202061803Ad102} as a cascade system.
The corresponding transformation is
\begin{equation} \label{wxh20202061804Ad102}
\begin{array}{l}
\begin{bmatrix}
I_n&S\\
0&I_m
\end{bmatrix}
\begin{bmatrix}
A +F_1C&BQ \\
-F_2C &G
\end{bmatrix}\begin{bmatrix}
I_n&S\\
0&I_m
\end{bmatrix}^{-1}\crr
=
\begin{bmatrix}
A +(F_1-SF_2)C & SG-[A +(F_1-SF_2)C ]S+BQ \\
-F_2C &G+F_2C S
\end{bmatrix},
\end{array}
\end{equation}
where $S\in     \R^{n\times m}$ is to be determined.
If we select $S$ properly such that
\begin{equation} \label{20208292004Ad102}
SG-[A +(F_1-SF_2)C ]S+BQ=0,
\end{equation}
then   the right side matrix of \dref{wxh20202061804Ad102} is Hurwitz if and only if the matrices $A+(F_1-SF_2)C$ and $G+F_2CS$ are Hurwitz.

\begin{theorem}\label{Th20201141911}

Suppose that system \dref{2020151017} is observable for  $\Omega(G)$ and $G\in \R^{m\times m}$ is known.
Then, there exist $F_1\in  \R^n  $,   $F_2\in  \R^m $  and $Q\in  \R^{1\times m}$   such that
$(G,Q)$ is observable and
 the solution of the observer \dref{20201021624} satisfies
\dref{wxh20202062106Ad102}.
 Moreover, $F_1$,  $F_2$ and $Q$ can be selected  by the following   scheme:
(a)~ Select  $F_0\in  \R^n  $ such that $A +F_0C $ is Hurwtiz,
select  $P \in  \R^{1\times m }$ such that  $(G,P)$ is observable  and
select $F_2$ such that
 $G+F_2P$ is Hurwitz;
 (b)~  Solve the   equations
   \begin{equation} \label{202010021741}
   (A +F_0C )S-SG=BQ\ \ \mbox{ and}\ \ CS=P
   \end{equation}
   to get $S\in  \R^{n\times m}$ and
 $Q\in  \R^{1\times m}$;
(c)~   Set $F_1=F_0+SF_2$.

 \end{theorem}
\begin{proof}
 Since system \dref{2020151017} is observable for  $\Omega(G)$,  it follows from Lemma \ref{lm20209132036} that
 $(A,C)$ is observable and the transmission zeros condition \dref{2020213905} holds.
   Therefore,
there exists an
    $F_0\in  \R^n $
  such that    $A +F_0C$  is    Hurwitz and
  \begin{equation} \label{20201202102}
  \sigma(A +F_0C)\cap\sigma(G)=\emptyset.
  \end{equation}
  By \cite{Rosenblum1956DMJ},  \dref{20201202102} and \dref{2020213905},
 the   equations
 \dref{202010021741}  admits a   solution
  $S\in  \R^{n\times m} $ and $Q\in \R^{1\times m}$.
Since $(G,P)$ is observable, there exists an $F_2$ such that
 $G+F_2P$ is Hurwitz. As a result, $F_1=F_0+SF_2$ is well defined.

 Now, we claim that $(G,Q)$ is observable. Indeed,
if we suppose that $ Gh=\lambda h$ and $Qh=0$ with $\lambda\in\sigma(G)$.
  Then,  the Sylvester equation in \dref{202010021741}
 turns out  to be $(A+F_0C -\lambda)Sh  =BQh=0$.
 By \dref{20201202102}, we conclude that    $\lambda\notin\sigma(A+F_0C)$ and hence $A+F_0C -\lambda$ is invertible. As a result, $Sh =0$ and $CSh=P h=0$.
   By \cite[p. 15, Remark 1.5.2]{Weissbook}, we can conclude $h=0$ due to the observability of
    $(G,P )$. Therefore, we obtain the  observable system $(G,Q)$ by which
 system \dref{2020151017} can be written as a cascade system \dref{202010021732} for some initial state.

 Let
\begin{equation} \label{20206231121}
\tilde{x} (t)=x (t)-\hat{x} (t) \ \ \mbox{and}\ \ \tilde{v}(t)=v(t)-\hat{v} (t).
\end{equation}
Then, the error between \dref{202010021732} and the observer \dref{20201021624}   is governed by
\begin{equation} \label{202010021800}
 \left\{\begin{array}{l}
\disp \dot{\tilde{x}} (t) = [A +(F_0+SF_2)C ]\tilde{x} (t)+B Q \tilde{v} (t),\crr
\disp  \dot{\tilde{v}} (t) =  G \tilde{v} (t) -F_2C\tilde{x}(t) .
\end{array}\right.
\end{equation}
Thanks to the choice of $F_1$ and $F_2$, a simple computation shows that  the following matrices are similar each other
 \begin{equation} \label{2020914721}
 \begin{bmatrix}
A+(F_0+SF_2)C& BQ  \\ -F_2C&  G
\end{bmatrix}   \ \ \mbox{and}\ \ \begin{bmatrix}
A +F_0C&0   \\
-F_2C&G+F_2P \\
 \end{bmatrix}.
\end{equation}
  Since the matrices  $A +F_0C$ and $G+F_2P$ are Hurwitz,
 both the matrices in \dref{2020914721} are Hurwitz. As a result,
 \dref{wxh20202062106Ad102}  holds due to      \dref{20206231121}.
 \end{proof}

  \begin{remark}\label{Re202010021803}

Equations \dref{202010021741} are   known as the regulator equations which are instrumental to establishing
the linear output regulation theory \cite{HuangJ2004}. The transmission zeros condition \dref{2020213905} can be represented as
\begin{equation}\label{20201031209}
 {\rm rank }\begin{bmatrix}
    A-\lambda& B \\
        C   &0
            \end{bmatrix}=n+m,\ \ \ \forall\ \lambda\in\sigma(G).
\end{equation}

\end{remark}
\section{Observer design with general disturbance}\label{Se.5}
The disturbance considered in Section \ref{KnownDynamics} is quite ideal.
In engineering applications, the disturbance dynamics are  usually unknown  or at least partially unknown.
It is therefore more realistic to consider the  disturbance  $d\in \mathbb{S}$.
 Suppose that we have known that $G$ is a
 ``rough approximation" of the
   dynamics of  $d$.
 In order to make use this  prior dynamics,
  we first  represent the disturbance dynamically as an output of a system dominated by $G$.

\begin{lemma} \label{lm20201211809}
Let $(G,Q)$ be  an observable system  with   state space $\R^{m+1}$ and output space $\R$.   Suppose that $0\in \sigma(G)$.
 Then,  for any
$d\in \mathbb{S } $, there exists  a  $v_0\in \R^{m+1}$ such that
\begin{equation}\label{20201212010}
  \left\{\begin{array}{l}
\disp \dot{ v} (t) =  Gv  (t)+\frac{B_{d}}{ QB_{d} }  \dot{e}(t),\ \ v (0)=v_{0},\crr
\disp d (t)=Qv (t),
\end{array}\right.
\end{equation}
where $B_{d}\in  \R^{m+1} $ is the eigenvector corresponding to the eigenvalue  $0$ of $G$,
 $e=(I-\mathbb{P}_G)  {d}$  and $\mathbb{P}_G$ is given by \dref{20201211811}.

\end{lemma}
\begin{proof}
Since $(G,Q)$ is observable and $ 0\in\sigma(G)$, it follows from the Hautus test \cite[p.15, Remark 1.5.2]{Weissbook} that ${\rm Ker} G\cap {\rm Ker}Q=\{0\}$.
 This means $ QB_{d} \neq 0$  by the fact $GB_{d}=0$ and $B_{d}\neq0$.  The first equation of \dref{20201212010} therefore  makes sense.

 Since $\mathbb{P}_G {d}\in\Omega(G)$, it can be written dynamically as
 \begin{equation}\label{20201211842}
  \left.\begin{array}{l}
\disp \dot{ {v}}_1(t) =  Gv_1  (t),\ \
\disp (\mathbb{P}_G {d})(t)=Qv_1(t)
\end{array}\right.
\end{equation}
for some initial state.
Since
$GB_{d}=0$, if we let
 \begin{equation}\label{20201212021}
v_2(t)=     \frac{d(t)-(\mathbb{P}_G {d})(t)}{ QB_{d} } B_{d}      =
 \frac{e(t)}{ QB_{d} } B_{d},\ \ t\geq0,
\end{equation}
 then
 \begin{equation}\label{20201212023}
  \left\{\begin{array}{l}
\disp \dot{ {v}}_2(t) =  Gv_2  (t)+\frac{B_{d}}{ QB_{d} } \dot{e}(t),\ \ v_2(0)=\frac{e(0)}{ QB_{d} } B_{d},\crr
\disp e(t)=Qv_2(t).
\end{array}\right.
\end{equation}
 System  \dref{20201212010} then follows from
 \dref{20201211842}  and  \dref{20201212023} by letting  $v(t)=v_1(t)+v_2(t)$.
  \end{proof}

%
%

Now, we design an observer for system \dref{2020151017} under the assumption that $d\in \mathbb{S}$.
 By Lemma \ref{Lm20201021046},
  we can suppose without loss of generality
  that $A,B$,  $C$, $G$ and $B_d$  satisfy the following assumptions:

\begin{assumption}\label{Assum1}
Let    $n$ and $m$ be positive integers, let the matrices $A$, $B$ and $C$  be given by  \dref{202010021047}
	 with $b_1=1$ and
  $b_2=b_3=\cdots=b_n=0$
and let the matrices  $G$, $E$ and $B_{d}$ be given by
 \begin{equation}\label{20209923729}
		G  =	\begin{bmatrix}
0&1&0&\cdots&0\\
		0&0&1&\cdots&0\\
		\vdots&\vdots&\vdots&\ddots&\vdots\\
		0&0&0&\cdots&1\\
		   0 & g_1 &  g_2 &\cdots&  g_m
		\end{bmatrix},\ \ E=\begin{bmatrix}
0\\
0\\
0\\
\vdots\\
1\\
		\end{bmatrix}\ \ \mbox{and}\ \  B_{d}=\begin{bmatrix}
1\\
0\\
0\\
\vdots\\
0\\
		\end{bmatrix}\in\R^{m+1},
		\end{equation}
where  $ g_j\in \R$,  $j= 1,2,\cdots,m$  such that
  \begin{equation} \label{2020923735}
  \sigma(G)\subset  \mathbb{C}_{+}.
  \end{equation}
\end{assumption}
   The assumption \dref{20209923729}    implies that $ 0\in\sigma(G)$ and thus  $GB_{d}=0$   for any $g_j$, $j=1,2,\cdots,m$.
Moreover, the dynamics dominated by $G$ of  \dref{2020923735} contain all signals of  harmonic, polynomial signals,
exponential signals  and their linear combinations.
   Inspired by  Theorem \ref{Th20201141911},
 the EDO  of system \dref{20201212010}   is designed as
  \begin{equation} \label{2020919191147}
  \left\{\begin{array}{l}\disp
\dot{ \hat{x}} (t) =
[A+( K_{\omega_o} +SE )C]\hat{x}(t)+B Q\hat{v}(t) -( K_{\omega_o} +SE )y(t)+Bu(t),\crr
\disp \dot{ \hat{v}} (t) =G\hat{v}(t)-E C \hat{x}(t)+
 E  y(t) ,
\end{array}\right.
 \end{equation}
where $G$, $E$,     $K_{\omega_{o}},  S $ and $Q$ are chosen by the following scheme:
\begin{itemize}

\item Choose $G$ and $E$    in terms of the Assumption \ref{Assum1} and the prior  information  about the disturbance (This will be considered in Sections \ref{ESO} and \ref{Se.7});

\item   Choose $K=[k_1-a_1\quad  k_2-a_2\quad  \cdots\quad k_n-a_n]^{\top}$  and
$P=[p_0\quad p_1-g_1\quad \cdots\quad p_m-g_m]$ such that the matrices
$ A + KC $ and  $G + E P$
are Hurwitz. Let
  \begin{equation} \label{2020927830B}
  K_{\omega_o}=[
 k_1\omega_o^n -a_1\ \
  k_2\omega_o^{n-1} -a_2\ \
  \cdots\ \
  k_n\omega_o-a_n]^{\top}
 \end{equation}
 and
 \begin{equation} \label{2020927837}
  P_{ \omega_{o}}=[
 p_0\omega_{o}^{m+1}\ \  p_1\omega_{o}^m-g_1\ \  \cdots\  \ p_m\omega_{o}-g_m],
  \end{equation}
 where  ${\omega_o}$ is a positive   tuning parameter;

 \item  Solve  the     equations
    \begin{equation} \label{2020861532regualtor}
   \left.\begin{array}{l}
\disp (A+K_{\omega_o} C )S-SG =BQ,\ \ \ \
   CS=P_{\omega_{o}},
\end{array}\right.
      \end{equation}
 to get $S\in  \R^{n\times (m+1)}$
 and $Q\in    \R^{1\times (m+1)} $.

 \end{itemize}

\begin{lemma}\label{Lm2020923809}
Under   Assumption \ref{Assum1} and the scheme of observer design, the equations \dref{2020861532regualtor} are always solvable. Moreover,  the following assertions are true:

(i) System $(G,Q)$ is observable;

(ii) For any $s\in \mathbb{C}_{+}$, there exist  two  positive constants  $C_K$  and $C_A$, independent of $\omega_o$ and $s$, such that
 \begin{equation}\label{20209231123}
\|[s-(A+K_{\omega_o}C)] ^{-1}B\|_{{\R^{n}}}\leq \frac{C_K}{\omega_o},
  \end{equation}
 \begin{equation}\label{20209241053}
\frac{1}{|QB_{d}|}
 =\frac{1}{   |k_1p_0| \omega_{o}^{n+m+1}}
  \end{equation}
and
  \begin{equation} \label{2020925920}
 \|C[s-(A+K_{\omega_o}C  )]^{-1}\|_{{\R^{n}}}\leq \frac{C_A}{\omega_o^n};
\end{equation}

(iii)  For any $s\in \mathbb{C}_{+}$,  there exists a  positive constant   $C_G$,    independent of $\omega_o$ and $s$,  such that
 \begin{equation}\label{202010031638}
\| [s-(G+EP_{\omega_{o}} )]^{-1} E\|_{\R^{m+1}}
  \leq \frac{C_G}{\omega_o^{m+1}},\ \ \forall\; s\in   \mathbb{C}_{+};
  \end{equation}

  (iv) If $G$ is diagonalizable, then  there exist two    positive constants     $C_S$  and $C_Q$,
independent of $\omega_o$ and $s$, such that
\begin{equation}\label{20201141033}
\| Sv\|_{\R^n}
  \leq  {C_S}\|v\|_{\R^{m+1}}{\omega_o^{m+n}},\ \ \forall\; v\in   \R^{m+1}
  \end{equation}
and
\begin{equation}\label{20201141436}
|Q[s-(G+EP_{\omega_o})]^{-1}B_d| \leq C_Q\omega_o^{m+n}, \ \  \forall\; s\in   \mathbb{C}_{+}.
  \end{equation}


 \end{lemma}
\begin{proof}
Since $A+KC$ is Hurwitz and by  the choice of
$A$, $B$, $C$ and $K_{\omega_o}$,  a simple computation shows that $A+K_{\omega_o}C$ is Hurwitz as well and
   \begin{equation} \label{2020923905}
 \lambda \omega_o\in \sigma(A+K_{\omega_o}C)\ \ \mbox{if and only if }\ \lambda  \in \sigma(A+ KC) .
 \end{equation}
Noting that the matrix
 \begin{equation}\label{20209231022}
	 A+K_{\omega_o}C  =		 \begin{bmatrix}
0&0& \cdots&0& k_1\omega_o^n\\
1&0& \cdots&0& k_{2}\omega_o^{n-1}\\
 0&1& \cdots&0& k_{3}\omega_o^{n-2}\\
\vdots&\vdots&\ddots&\vdots&\vdots \\
		0&0& \cdots&1& k_n\omega_o \\
		\end{bmatrix}
\end{equation}
is Hurwitz,
  \dref{2020923735} implies  that
 \begin{equation} \label{2020923823}
 \sigma(A+K_{\omega_o}C)\cap\sigma(G)=\emptyset.
   \end{equation}
Moreover,
a simple computation shows  that
\begin{equation}\label{20209231109}
[\lambda-(A+K_{\omega_o}C)]^{-1}B
  =		\disp \frac{ \mathcal{K}_{\lambda}}{\rho_A (\lambda,\omega_o)},\ \ \lambda\in \mathbb{C}_{+},
\end{equation}
where
\begin{equation} \label{202029231453K}
\mathcal{K}_{\lambda}=\begin{bmatrix}
   k_n\lambda ^{n-2}\omega_{o} + k_{n-1}\lambda ^{n-3}\omega_{o}^2 +\cdots+ k_3\lambda \omega_{o}^{n-2} + k_2\omega_{o}^{n-1} -\lambda ^{n-1}  \\
k_n\lambda ^{n-3}\omega_{o} + k_{n-1}\lambda ^{n-4}\omega_{o}^2 + \cdots+k_3\omega_{o}^{n-2}-\lambda ^{n-2}\\
\vdots\\
  k_n\omega_{o}-\lambda  \\
		-1 \\
		\end{bmatrix}
\end{equation}
and
\begin{equation} \label{202029231453}
\rho_{A}(\lambda,\omega_o)=k_n\lambda ^{n-1}\omega_{o} + k_{n-1}\lambda ^{n-2}\omega_{o}^2 +\cdots +k_2 \lambda  \omega_{o}^{n-1}+   k_1\omega_{o}^n-\lambda ^n  .
\end{equation}
 Hence,  the following  transmission zeros condition holds:
   \begin{equation} \label{20202923823}
C[\lambda-(A+K_{\omega_o}C)]^{-1}B=\frac{-1}{ \rho_{A}(\lambda,\omega_o) }\neq 0  ,\ \ \forall\ \lambda\in  \mathbb{C}_{+}.
\end{equation}
By \cite{Natarajan2016TAC},  \dref{2020923823}  and  \dref{20202923823},
the equations \dref{2020861532regualtor}  are solvable.

(i).  Suppose that $ Gh=\lambda h$ and $Qh=0$ with $\lambda\in\sigma(G)$.
  Then,  the Sylvester equation in \dref{2020861532regualtor}
 turns out  to be $(A+K_{\omega_o}C -\lambda)Sh  =BQh=0$.
 By \dref{2020923823}, we conclude that    $\lambda\notin\sigma(A+K_{\omega_o}C)$ and hence $A+K_{\omega_o}C -\lambda$ is invertible. As a result, $Sh =0$ and $CSh=P_{\omega_{o}}h=0$.
   By \cite[p. Remark 1.5.2]{Weissbook}, we can conclude $h=0$ provided
    $(G,P_{\omega_{o}})$ is observable. Using   \cite[p. Remark 1.5.2]{Weissbook} again,
    $(G,Q)$ is observable
      if we can prove
    $(G,P_{\omega_{o}})$ is observable.
    Actually,  for any $ Gv=\lambda v$ and $P_{\omega_{o}}v=0$ with $\lambda\in\sigma(G)$,
  we have
  $(G+EP_{\omega_{o}})v=Gv=\lambda v$.   Since $G+EP$ is Hurwitz, it follows from
    \dref{20209923729} and \dref{2020927837}  that  $G+EP_{\omega_{o}}$ is Hurwitz
     as well and
     \begin{equation} \label{2020923904}
 \lambda \omega_o\in \sigma(G+EP_{\omega_{o}})  \mbox{ if and only if }  \lambda  \in \sigma(G+EP) .
 \end{equation}
 By \dref{2020923735} and the fact $\lambda\in\sigma(G)$, we obtain $\lambda\notin \sigma(G+EP_{\omega_{o}})$.
Hence,   $(G+EP_{\omega_{o}})v= \lambda v$   implies that  $v=0$.  By \cite[p. Remark 1.5.2]{Weissbook},
$(G,P_{\omega_{o}})$ is observable.

(ii). Since $A+K_{\omega_o}C$ is Hurwitz, it follows from \dref{20209231022}  that $k_1\neq0$.
Hence, \dref{20209231123} can be obtained by
\dref{20209231109}, \dref{202029231453K} and \dref{202029231453} easily.
 Noting the $GB_{d}=0$,   it follows from  \dref{2020861532regualtor} that
   \begin{equation}\label{2020923918}
   (A+K_{\omega_o}C  )SB_{d}  =BQB_{d}
 \end{equation}
  and hence
 \begin{equation}\label{2020923910}
P_{\omega_{o}}B_{d}=CSB_{d}=C(A+K_{\omega_o}C  )^{-1}BQB_{d}.
  \end{equation}
  By  \dref{20209923729} and \dref{2020927837},
  \begin{equation}\label{20209241846}
P_{\omega_{o}}B_{d}= p_0\omega_o^{m+1}.
  \end{equation}
  Since $(G,Q)$ is observable and $GB_{d}=0$,   the Hautus test
 \cite[p.15, Remark 1.5.2]{Weissbook} implies that $QB_{d}\neq0$.
  Hence, we combine   \dref{20202923823},  \dref{20209241846} and \dref{2020923910} to obtain $ p_0\omega_o^{m+1}\neq0$ and
\begin{equation}\label{20209231429}
\frac{1}{QB_{d}}=
  \frac{C(A+K_{\omega_o}C  )^{-1}B}{P_{\omega_{o}}B_{d}}=-\frac{1}
{   k_1  p_0  \omega_{o}^{n+m+1 }},
  \end{equation}
which leads to \dref{20209241053}  easily.
 In view of \dref{20209231022}, a  straightforward computation shows that
   \begin{equation} \label{20209251100}
C[s-(A+K_{\omega_o}C)]^{-1} =\frac{ -1}{ \rho_{A}(s,\omega_o) } [
 1\quad s\quad s^2\quad \cdots\quad s^{n-1}],\ \forall\; s\in \mathbb{C}_{+},
\end{equation}
which, together with \dref{202029231453}, leads to  \dref{2020925920}  easily.

(iii).
 By  a straightforward computation, it follows that
\begin{equation}\label{20209241057}
	G+EP_{\omega_{o}} =	\begin{bmatrix}
0&1&0&\cdots&0\\
		0&0&1&\cdots&0\\
		\vdots&\vdots&\vdots&\ddots&\vdots\\
		0&0&0&\cdots&1\\
		  p_0\omega_o^{m+1}& p_1\omega_o^{m}&  p_2\omega_o^{m-1}&\cdots&  p_m\omega_o
		\end{bmatrix}
\end{equation}
and hence
\begin{equation}\label{20209241100}
		[s-(G+EP_{\omega_{o}})]^{-1}E=\frac{-1}{ \rho_G(s,\omega_o)} \begin{bmatrix}
  1 \\
    s \\
    \vdots\\
  s^{m}
\end{bmatrix},\ \ s\in \mathbb{C}_{+},
\end{equation}
where
\begin{equation}\label{20209242002}
		\rho_G(s,\omega_o)=p_0 \omega_o^{m+1} +p_1\omega_o^{m} s + \cdots+ p_m\omega_o s^{m} - s^{m+1} .
\end{equation}
Since $G+EP_{\omega_{o}}$ is Hurwitz, we have $p_0\neq0$ and hence \dref{202010031638} follows from \dref{20209241100} and \dref{20209242002}.

(iv).
Since $G$ is diagonalizable, for any $v\in {\R^{m+1}}$, there exists a sequence  $v_0,v_1,v_2,\cdots,v_m$ such that $v=\sum_{j=0}^{m}v_j\varepsilon_j$, where
 $G\varepsilon_j=\lambda_j\varepsilon_j$ with $\lambda_j\in\sigma (G)$, $j=0,1,2,\cdots,m$.
By     \dref{2020861532regualtor},  \dref{20209231109} and \dref{20202923823}, we have
 \begin{equation}\label{20209231442}
S    \varepsilon_j=(A+K_{\omega_o}C -\lambda_j)^{-1}BQ\varepsilon_j= -\frac{ \mathcal{K}_{\lambda j}}{\rho_{A} (\lambda_j,\omega_o)}Q\varepsilon_j
  \end{equation}
and
\begin{equation}\label{20209231444}
P_{\omega_{o}}\varepsilon_j=CS    \varepsilon_j=C(A+K_{\omega_o}C -\lambda_j)^{-1}BQ\varepsilon_j= \frac{ Q\varepsilon_j }{ \rho_A (\lambda_j,\omega_o) } .
  \end{equation}
Consequently,
\begin{equation}\label{20209231446}
Q\varepsilon_j=     P_{\omega_{o}} \varepsilon_j \rho_A(\lambda_j,\omega_o) ,\ \ j=0,1,\cdots,m,
  \end{equation}
which, together with \dref{20209231442}, gives
 \begin{equation}\label{20209231449}
Sv=\sum_{j=0}^{m}v_jS\varepsilon_j=  - \sum_{j=0}^{m}v_j     P_{\omega_{o}}\varepsilon_j \mathcal{K}_{\lambda j}   .
  \end{equation}
Combining  \dref{202029231453K}, \dref{2020927837} and \dref{20209231449}, we obtain \dref{20201141033} easily.


Taking \dref{20209241057}  and \dref{20209923729} into account, a simple computation shows that
\begin{equation}\label{20201141456}
 [s-(G+EP_{\omega_o})]^{-1}B_d=\frac{1}{\rho_G(s,\omega_o)}
 \begin{bmatrix}
  p_ms^{m-1}\omega_o+p_{m-1}s^{m-2}\omega_o^2+\cdots+p_1\omega_o^m-s^m \\
    -p_0\omega_o^{m+1} \\
    -p_0\omega_o^{m+1}s\\
    \vdots\\
 -p_0\omega_o^{m+1} s^{m-1}
\end{bmatrix}
\end{equation}
for any $ s\in \mathbb{C}_{+}$.
By \dref{20209242002} and the fact $p_0\neq0$,  there exists a positive constant $M_1$,   independent of $\omega_o$ and $s$,  such that
\begin{equation}\label{20201141516}
  \|G[s-(G+EP_{\omega_o})]^{-1}B_d\|_{\R^{m+1}}\leq M_1\|G\| ,\ \ \forall \ s\in \mathbb{C}_{+}.
  \end{equation}
Consequently, it follows from \dref{20201141033},  \dref{2020925920}  and \dref{20201141516} that
\begin{equation} \label{202011041521}
  |C[s-(A+K_{\omega_o}C  )]^{-1} S G[s-(G+EP_{\omega_o})]^{-1}B_d|\leq C_SM_1\|G\| {C_A}{\omega_o^m}.
\end{equation}
 Combing \dref{2020927837}, \dref{20209242002} and \dref{20201141456},
there exists a positive constant $M_2$,   independent of $\omega_o$ and $s$,  such that
\begin{equation}\label{20201141505}
  |P_{\omega_o}[s-(G+EP_{\omega_o})]^{-1}B_d|\leq M_2\omega^m_{o},\ \ \forall \ s\in \mathbb{C}_{+}.
  \end{equation}

For any $v\in \R^{m+1}$, it follows from \dref{2020861532regualtor}  and \dref{20202923823} that
\begin{equation}\label{20201141134}
Qv =\frac{P_{\omega_o}v- C(A+K_{\omega_o}C)^{-1}SGv}{C(A+K_{\omega_o}C)^{-1}B} =
-\rho_{A}(\lambda,\omega_o) [P_{\omega_o}v- C(A+K_{\omega_o}C)^{-1}SGv].
  \end{equation}
As a result, there exists an $M_3$,   independent of $\omega_o$,  such that
 \begin{equation}\label{20201141530}
|Qv| =M_3
 \omega_o^n [|P_{\omega_o}v|+|C(A+K_{\omega_o}C)^{-1}SGv|],\ \ \forall\ v\in \R^{m+1},
  \end{equation}
  which, together with  \dref{202011041521} and \dref{20201141505}, leads to
  \dref{20201141436} easily.
   \end{proof}

\begin{theorem}\label{Th20209182115}
Under  Assumption \ref{Assum1},
 the EDO    \dref{2020919191147} of \dref{2020151017}  is well-posed:  For any    $d\in  \mathbb{S}  $,
  $( \hat{x}(0),\hat{v} (0)) \in
  {\R^{n}}\times {\R^{m+1}}$   and   $u\in L^2_{\rm loc}[0,\infty) $,
     there exists a positive constant $M_1$, independent of $\omega_o$, such that
    \begin{equation} \label{2020261018**}
 \disp\lim_{t\to\infty}\| (x (t)-\hat{x} (t) ,v(t)-\hat{v} (t))\|_{{\R^{n}}\times {\R^{m+1}} }
  \leq \frac{M_1\|  e \|_{\mathbb{S}} }{\omega_o} ,
      \end{equation}
        where  $e=(I-\mathbb{P}_G)  {d}$  and $\mathbb{P}_G$ is given by \dref{20201211811}.
   In particular, there exists a positive constant $M_2$,  independent of $\omega_o$, such that
      \begin{equation} \label{2020623959}
 \disp  \lim_{t\to\infty} |d(t)- Q\hat{v} (t) |\leq \frac{M_2\|e\|_\mathbb{S}}{ \omega_o}.
      \end{equation}

 \end{theorem}
\begin{proof}
By Lemma \ref{lm20201211809}, system \dref{2020151017}
 can be written dynamically  as
 \begin{equation}\label{20201131610}
  \left\{\begin{array}{l}
  \disp \dot{x} (t) = A  x (t)+B [Qv (t)+ u(t)] ,\
\disp y(t)=Cx (t),\crr
\disp \dot{ v} (t) =  Gv  (t)+\frac{B_{d}}{ QB_{d} }  \dot{e}(t),
\end{array}\right.
\end{equation}
  where
$e=(I-\mathbb{P}_G)  {d}$  and $\mathbb{P}_G$ is given by \dref{20201211811}.
 Let
\begin{equation} \label{20209191511}
\tilde{x} (t)=x (t)-\hat{x} (t) \ \ \mbox{and}\ \ \tilde{v}(t)=v(t)-\hat{v} (t).
\end{equation}
Then, the error is governed by
\begin{equation} \label{20206231124}
 \left\{\begin{array}{l}
\disp \dot{\tilde{x}} (t) = [A +( K_{\omega_o} +SE)C ]\tilde{x} (t)+B Q \tilde{v} (t),\crr
\disp  \dot{\tilde{v}} (t) =  G \tilde{v} (t) -EC\tilde{x}(t)+  \frac{B_{d}}{QB_{d}}\dot{e}(t).
\end{array}\right.
\end{equation}
System \dref{20206231124} can be written as
\begin{equation} \label{20206242019}
  \frac{d}{dt}(\tilde{x}  (t),  \tilde{v}  (t) )^{\top}=\A(\tilde{x}  (t),  \tilde{v}  (t) )^{\top}+
  \mathcal{B} {\dot{e}(t)} ,
\end{equation}
where
\begin{equation}\label{20209242021}
 \A=\begin{bmatrix}
     A +( K_{\omega_o} +SE)C& BQ\\
    -EC&  G
    \end{bmatrix}\ \ \mbox{and}\ \ \mathcal{B}=
   \frac{1}{QB_{d}} \begin{bmatrix}
   \disp   0\crr
    \disp B_{d}
    \end{bmatrix} .
  \end{equation}
 In terms of the solution $S$ of the Sylvester equation of \dref{2020861532regualtor}, we introduce the transformation
 \begin{equation} \label{2020142058}
\begin{bmatrix}
\check{x}(t)\\
\check{v}(t)
\end{bmatrix}=
\mathbb{P}\begin{bmatrix}
\tilde{x}(t)  \\
\tilde{v}(t)
\end{bmatrix}, \ \ \mathbb{P}=\begin{bmatrix}
I_{ {n}} &S\\
0&I_{ {m+1}}
\end{bmatrix}.
 \end{equation}
Thanks to the choice of $ K_{\omega_o} ,E$ and $S$, system     \dref{20206231124}  can be converted into
   the following system:
\begin{equation} \label{20191161439}
 \left\{\begin{array}{l}
\disp \dot{\check{x}} (t) =( A+ K_{\omega_o} C)  \check{x}(t) + \frac{SB_{d}}{QB_{d}}\dot{e}(t),\crr
\disp  \dot{\check{v}} (t) =   (G+EP_{\omega_{o}} )  \check{v}(t)-EC\check{x}(t)+\frac{B_{d}}{QB_{d}}\dot{e}(t).
\end{array}\right.
\end{equation}
We denote the system matrix and the input matrix of \dref{20191161439}   by
  \begin{equation}\label{20209232209}
 \A_S=\begin{bmatrix}
      A+ K_{\omega_o} C& 0\\
    -EC&  G+EP_{\omega_{o}}
    \end{bmatrix}\ \ \mbox{and}\ \ \mathcal{B}_S=
  \frac{1}{QB_{d}}\begin{bmatrix}
   \disp    SB_{d}\crr
    \disp B_{d}
    \end{bmatrix} .
  \end{equation}
  By a simple computation, it follows that
  \begin{equation}\label{20209242027}
 \mathbb{P}\A \mathbb{P}^{-1}=\A_S\ \ \mbox{and}\ \ \mathcal{B}_S=\mathbb{P}\mathcal{B},
  \end{equation}
  where the Sylvester equation in \dref{2020861532regualtor}  has been used.
 For any $s\in  \mathbb{C}_{+}$, a simple computation shows that
\begin{equation} \label{2020922838}
\begin{array}{l}
\disp \mathbb{P}^{-1}(s-\A_S)^{-1}\mathcal{B}_S
\disp = \frac{1}{QB_{d}}
\begin{bmatrix}
    \disp   [s-(A+K_{\omega_o}C  )]^{-1} SB_{d}+S J(s)\crr
    \disp -J(s)
        \end{bmatrix},
 \end{array}
\end{equation}
where
 \begin{equation}\label{202010031547}
J(s)=[s-(G+EP_{\omega_{o}} )]^{-1}EC [s-(A+K_{\omega_o}C  )]^{-1} SB_{d}-
[s-(G+EP_{\omega_{o}} )]^{-1} B_{d}  .
\end{equation}

Noting the $GB_{d}=0$,   it follows from  \dref{2020861532regualtor} that
   \begin{equation}\label{20209242033}
   SB_{d}  =(A+K_{\omega_o}C  )^{-1}BQB_{d},
 \end{equation}
 and  hence
\begin{equation}\label{20209242034}
\begin{array}{rl}
 \disp  \frac{ [s-(A+K_{\omega_o}C  )]^{-1} SB_{d} }{QB_{d}} &\disp =[s-(A+K_{\omega_o}C  )]^{-1}   (A+K_{\omega_o}C  )^{-1}B \crr
  &\disp =   (A+K_{\omega_o}C  )^{-1}[s-(A+K_{\omega_o}C  )]^{-1}B.
  \end{array}
 \end{equation}
By  Lemma \ref{Lm2020923809}, there exist  two  positive constants  $ C_K $  and $C_S$ such that
\begin{equation}\label{20209242039}
  \left\|\frac{ [s-(A+K_{\omega_o}C  )]^{-1} SB_{d} }{QB_{d}}\right\|_{\R^{n}}  \leq \frac{C_K}{\omega_o^2}
 \end{equation}
 and
 \begin{equation}\label{20201031611}
\frac{\|SJ(s)\|_{\R^n}}{|QB_d|}\leq
  \frac{C_S\|J(s)\|_{\R^{m+1}}  \omega_{o}^{m+n}}{    \omega_{o}^{n+m+1}}=
  \frac{C_S\|J(s)\|_{\R^{m+1}}  }{    \omega_{o} },\ \ \forall\ s\in \mathbb{ C}_{+}.
\end{equation}
 By   \dref{20209242002}, \dref{20201141456}  and the fact $p_0\neq0$, there
 exists  a     $ C_J>0 $    such that
  \begin{equation}\label{20209251111}
  \|J(s)\|_{\R^{m+1}} <C_J,\ \ \forall\  s\in \mathbb{ C}_{+}.
 \end{equation}
 Combing \dref{20201031611},  \dref{20209251111}, \dref{202010031547},
 \dref{20209242039},  \dref{2020922838} and \dref{20209241053},    we arrive at
\begin{equation} \label{20209242046}
\|\mathbb{P}^{-1}(s-\A_S)^{-1}\mathcal{B}_S\|_{ \R^{n} }\leq \frac{C_\A}{\omega_o },\  \ \forall\ s\in \mathbb{C}_{+},
\end{equation}
 where $C_\A$ is a positive constant independent of $\omega_o$ and $s$.
 Furthermore, it follows from  \dref{20209242027} that
 \begin{equation} \label{20209242125}
\| (s-\A )^{-1}\mathcal{B} \|_{\R^{n}}=\|\mathbb{P}^{-1}(s-\A_S)^{-1}\mathcal{B}_S\|_{ \R^{n} }\leq \frac{C_\A}{\omega_o },\  \ \forall\ s \in \mathbb{C}_{+}.
\end{equation}
Since both $A+ K_{\omega_o} C$
and $ G+B_{d}P_{\omega_{o}}$ are Hurwitz and satisfy \dref{2020923904} and \dref{2020923905}, respectively,
the operator $\A$ is also Hurwitz.
By virtue of  Lemma \ref{Lm20209172011}, there exists an $L_{\mathcal{B}}>0$, independent of $\omega_o$,  such that
 \begin{equation} \label{20209242129}
 \|e^{\A t}\mathcal{B}\|_{{\R^{n}}\times {\R^{m+1}}}\leq L_{\mathcal{B}}e^{-\omega_o t},\ \ t\geq0.
 \end{equation}
We solve \dref{20206242019} to obtain
 \begin{equation} \label{20206242130}
 \begin{array}{l}
 \disp \|(\tilde{x}  (t),  \tilde{v}  (t) ) \|_{{\R^{n}}\times {\R^{m+1}}}=\left\|e^{\A t}(\tilde{x}  (0),  \tilde{v}  (0) )^{\top}+
  \int_0^te^{\A (t-s)}\mathcal{B} \dot{e}(s)  ds\right\|_{{\R^{n}}\times {\R^{m+1}}}\crr
  \hspace{1cm}\disp \leq L_{\A}e^{-\omega_o t}\left\|(\tilde{x}  (0),  \tilde{v}  (0) ) \right\|_{{\R^{n}}\times {\R^{m+1}}}+
  L_{\mathcal{B}}\int_0^te^{-\omega_o (t-s)} \frac{ \|\dot{e}\|_{\infty}}{\omega_o }
    ds \crr
     \hspace{1cm}\disp \leq L_{\A}e^{-\omega_o t}\left\|(\tilde{x}  (0),  \tilde{v}  (0) ) \right\|_{{\R^{n}}\times {\R^{m+1}}}+
    \frac{ \| {e}\|_{\mathbb{S}}L_{\mathcal{B}}}{\omega_o  },
       \end{array}
\end{equation}
where $L_{\A}$ is a positive constant. This leads to \dref{2020261018**} from  \dref{20209191511}.

Now, we prove \dref{2020623959}. For any $s\in \mathbb{C}_{+}$,
it follows from \dref{2020922838} that
\begin{equation} \label{20201031553}
\disp \mathcal{Q}\mathbb{P}^{-1}(s-\A_S)^{-1}\mathcal{B}_S
\disp = -\frac{QJ(s)}{QB_{d}}
 ,\ \ \mathcal{Q}=(0,Q).
\end{equation}
By \dref{202010031638},  \dref{2020925920} and \dref{20201141033},  there exists an $M_4>0$ such that
\begin{equation} \label{20201141620}
\|[s-(G+EP_{\omega_{o}} )]^{-1}EC [s-(A+K_{\omega_o}C  )]^{-1} SB_{d}\|_{\R^{m+1}}\leq \frac{M_4}{\omega_o},\ \ \forall\ s\in \mathbb{C}_{+}.
\end{equation}
By \dref{20201141530}, \dref{2020927837},  \dref{2020925920} and \dref{20201141033},  there exists an $M_5>0$ such that
 \begin{equation}\label{20201141628}
|Qv| =M_5
 \omega_o^{n+m+1}\|v\|_{\R^{ m+1}}  ,\ \ \forall\ v\in \R^{m+1}.
  \end{equation}
We combine \dref{20201141620} and \dref{20201141628} to get
\begin{equation} \label{20201141634}
 \left|Q[s-(G+EP_{\omega_{o}} )]^{-1}EC [s-(A+K_{\omega_o}C  )]^{-1} SB_{d}\right | \leq M_4M_5\omega_o^{n+m},\ \ \forall\ s\in \mathbb{C}_{+},
\end{equation}
which, together with \dref{20201141436},  \dref{202010031547},   \dref{20209241053} and \dref{20201031553}, leads to
\begin{equation} \label{2020925843}
\begin{array}{l}
 \disp  \left|\mathcal{Q}\mathbb{P}^{-1}(s-\A_S)^{-1}\mathcal{B}_S\right|=   \left|\frac{QJ(s)}{|QB_{d}|}\right|
    \disp\leq\frac{ M_6}{\omega_o },
        \end{array}
\end{equation}
where $M_6$  is a positive constant  independent of $\omega_o$ and $s$.
 Owing to \dref{20209242027},
 we arrive at
\begin{equation} \label{20209251146}
\begin{array}{l}
 \disp|\mathcal{Q} (s-\A )^{-1}\mathcal{B} |  \leq \frac{ M_6}{\omega_o}, \ \ \forall\ s\in \mathbb{C}_{+}.
        \end{array}
\end{equation}
Similarly to \dref{20209232121}, we
apply  the inverse Laplace transform on \dref{20209251146} to obtain
\begin{equation}\label{20209251949}
\begin{array}{ll}
\disp \lim_{\omega_o\to\infty}|\mathcal{Q}e^{\A t}\mathcal{B}|
  \disp &=\disp
 \frac{1}{2\pi i}\lim_{\omega_o\to\infty}\lim_{T\to\infty}\int_{\gamma-iT}^{\gamma+iT}
 e^{st}\mathcal{Q}(s-\A )^{-1}\mathcal{B}ds\crr
   \hspace{1cm}&=\disp
 \frac{1}{2\pi i}\lim_{T\to\infty}\int_{\gamma-iT}^{\gamma+iT}
 e^{st}\lim_{\omega_o\to\infty} \mathcal{Q}(s-\A)^{-1} \mathcal{B}ds\crr
  &=\disp
 \frac{1}{2\pi i}\lim_{T\to\infty}\int_{\gamma-iT}^{\gamma+iT}
 e^{st}0ds=0
  ,\ \  t\geq0,
  \end{array}
\end{equation}
where $\gamma$ is a real number so that the contour path of the integration is in the region of convergence of $\mathcal{Q}e^{\A t}\mathcal{B}$.
Since $\A$ is Hurwitz with $\Lambda_{\max}(\A)=-\omega_o$, \dref{20209251949} implies that
\begin{equation}\label{20209252020}
|\mathcal{Q}e^{\A t}\mathcal{B}|\leq L_Qe^{-\omega_ot},\ \ t\geq0,
\end{equation}
where $L_Q$ is a positive constant   which is independent of $\omega_o$.
As a result, the solution of system \dref{20206231124} satisfies
 \begin{equation} \label{2020923929}
 \begin{array}{ll}
\disp\lim_{t\to\infty}|Q\tilde{v}(t)|&\disp =\lim_{t\to\infty}\left|\mathcal{Q}e^{\A t}(\tilde{x}(0),\tilde{v}(0))^{\top}\right|+\lim_{t\to\infty}
\left|\mathcal{Q}\int_0^te^{\A s}\mathcal{B}\dot{e}(t-s)ds \right|\crr
&\disp \leq  \lim_{t\to\infty}
\left|\int_0^tL_Qe^{-\omega_os}  |\dot{e}(t-s)|ds \right|\leq \frac{L_Q\|\dot{e}\|_{\infty}}{\omega_o},
\end{array}
\end{equation}
which, together with \dref{2020251843}, \dref{20209191511} and \dref{20201212010},   leads to \dref{2020623959}.
 \end{proof}

\begin{remark}\label{Re20209252208}
Suppose that  $\{\varepsilon_j\}_{j=0}^{m}$ is a sequence of eigenvectors
corresponding to the eigenvalues $\lambda_j$ of $G$, which  forms a basis for ${\R^{m+1}}$.
Then, for
any $v=[v_0\  v_1\  \cdots\  v_{m}]\in {\R^{m+1}}$,
\begin{equation} \label{20209261058}
v=[\varepsilon_0\  \varepsilon_1\  \cdots\  \varepsilon_{m}][\varepsilon_0 \  \varepsilon_1\  \cdots\  \varepsilon_{m}]^{-1}\begin{bmatrix}
v_0\\v_1\\ \vdots\\v_m\end{bmatrix}.
\end{equation}
By \dref{20209231444},  we obtain the analytic  expression of $Q$ as
\begin{equation} \label{20209261105}
Qv \disp =[Q\varepsilon_0\  Q\varepsilon_1\  \cdots\  Q\varepsilon_{m}][\varepsilon_0\  \varepsilon_1\  \cdots\  \varepsilon_{m}]^{-1}
\begin{bmatrix}
  v_{0} \\
  v_{1} \\
  \vdots \\
  v_{ m   }
\end{bmatrix}
\end{equation}
with
\begin{equation} \label{202010031648}
Q\varepsilon_j=\frac{  P_{\omega_{o}}  \varepsilon_j  }{{C(A+K_{\omega_o}C -\lambda_j)^{-1}B}},\ \
j=0,1,2,\cdots,m.
\end{equation}
Therefore, we can obtain the
 parameter   of the observer \dref{2020919191147} explicitly via
 \dref{202010031648} and  \dref{20209231449}.

 %

%
%
%
%
%
%
\end{remark}

\begin{remark}\label{Re20201222110}

    By \dref{2020261018**},  the accuracy of the   observer
     depends both on the optimal   approximation  of $d$ on $\Omega(G)$   and the decay rate $\omega_o $.
 From  this perspective, we need to choose
 $G$ such that $\Omega(G)$ is as large as possible so  that  the approximation error can be
 as small as possible. The choice of $G$ depends on the prior information about the disturbance.
 The more the prior information we have, the higher the steady-state error will be.
 In particular,  if we have known all the dynamics of the disturbance, i.e.,
 we have known $d\in \Omega(G)$ with known $G$,
   the    steady-state error of the   observer  \dref{2020919191147} becomes zero.
 Another way to improve the observer accuracy is to increase the gain $\omega_o$.   However,   the large $\omega_o   $ may lead to
peaking phenomenon  in    transient response and hence   it  may not be  feasible to  improve the accuracy
by increasing  $\omega_o $ only.
Hence, one of the  contributions  of the present work is  giving  a new way to improve the
 accuracy of the observer without
 increasing the high-gain $\omega_o$.

\begin{remark}\label{Re20202131041}
When the  error of approximation $\mathbb{P}_G$ is zero,  the system matrix
of the  error system
\dref{20206231124} is similar to the matrix $\A_S$ in \dref{20209232209}. Owing to the block-trigonal structure of $\A_S$, the poles of the  error system
\dref{20206231124} can be  assigned arbitrarily by adjusting $K_{\omega_o}$ and $P_{\omega_o}$. This means that  the prior information about the
control plant and disturbance can be fully used.
 Moreover, the control plant considered in this paper is wider  than
 the canonical form of
\cite{GAO2003} or \cite{KhalilTAC2008}.  Although this canonical form  can be
extended  technically by using  high-gain
  \cite{GUOZhaoSCL2011}, there still
  exists a waste of system prior  information.
  In fact,   only  some boundedness of the elements of
system matrix $A$ was used rather than the matrix itself. As a result of this,
the   poles of the observer error  system   without the external disturbance
  cannot  be  assigned arbitrarily.


\end{remark}

\end{remark}

\section{Extended dynamic observer with constant dynamics  }\label{ESO}
 In this section, we consider  the EDO \dref{2020919191147} with constant dynamics $G=0$.
It  adapts to the   worst  situations where  we  have  nothing prior  information  about the disturbance dynamics excepted
some boundedness.  For simplicity,    we only consider,  without loss of the generality,  the
following second order control plant, i.e.,
\begin{equation}\label{2020918849n=2}
		A =\begin{bmatrix}
0 & a_1\\
1 & a_{2}\\
		\end{bmatrix},  \ \ B=
\begin{bmatrix}
 1\\0	\end{bmatrix}\ \  \mbox{and}\ \ C=[0\   1] ,
		\end{equation}
where $a_j  \in\R$, $j=1,2 $.
   If we choose $G=0$, then
 it follows from
  \dref{201912311058} that
\begin{equation}\label{20191232039}
\begin{array}{l}
\disp \Omega   (G)  = \Big{\{}  v(t)  \ \Big{|}\
 \dot{v}(t)=0,\ v(0)\in \R,\
t\in   \R \Big{\}}=
\Big{\{} v(t)\equiv v(0)  \ \Big{|}\
   v(0)\in \R,\
t\in   \R \Big{\}},
\end{array}
\end{equation}
 which implies that the optimal approximation $\mathbb{P}_Gd$  of $d$ on $\Omega   (G)$    satisfies
$ \|(I-\mathbb{P}_G)d \|_{\mathbb{S}} = \|\dot{d}(t)\|_{\infty}$.
  We choose
  \begin{equation} \label{20209271116}
  B_d=E=1,\ \ K_{\omega_0}=[ \alpha_1 \omega_o^2-a_1  \ \    \alpha_2  \omega_o-a_2 ]^{\top} \ \
  \mbox{and}\ \  P_{\omega_o}= \alpha_3\omega_o
  \end{equation}
 such that $\alpha_3<0$ and the following matrix is Hurwitz:
  \begin{equation} \label{20209271122}
  \begin{bmatrix}
 0&\alpha_1\\
1& \alpha_2
  \end{bmatrix} .
  \end{equation}
  We solve the   equations  \dref{2020861532regualtor} to get
  \begin{equation} \label{20201161056}
S =(A   +K_{\omega_o}C)^{-1}B= [ -\disp  \alpha_2 \alpha_3\omega_{o}^2  \quad \disp  \alpha_3\omega_o
				]^{\top}  \mbox{ and } Q=\alpha_1\alpha_3\omega_o^3.
\end{equation}
In view of \dref{2020919191147}, the observer of system \dref{2020151017} with setting \dref{2020918849n=2} is found to be
\begin{equation} \label{20209272056}
  \left\{\begin{array}{l}\disp
\disp \dot{ \hat{x}}_1 (t) =a_1\hat{x}_2(t)+\alpha_1\alpha_3\omega_o^3\hat{v}(t)
-[(\alpha_1-\alpha_2\alpha_3)\omega_o^2-a_1][y(t)- \hat{x}_2(t)]+u(t),\crr
 \disp \dot{ \hat{x}}_2 (t) =\hat{x}_1(t)+a_2\hat{x}_2(t)
-[(\alpha_2+\alpha_3)\omega_o -a_2][y(t)- \hat{x}_2(t)] ,\crr
\disp \dot{ \hat{v}} (t) =[y(t)- \hat{x}_2(t)]
 .
\end{array}\right.
 \end{equation}

 In order to make a comparison to ESO \cite{GAO2003}  and the high-gain observer \cite{KhalilTAC2008}, we consider system
 \begin{equation} \label{20209272135}
 \left\{\begin{array}{l}
 \disp \dot{z}(t)=A^{\top}z(t)+C^{\top}[d(t)+u(t)],\crr
 y(t)=B^{\top}z(t),
 \end{array}\right.\ \ z(t)=[z_1(t)\quad z_2(t)]^{\top},
\end{equation}
 where $A,B $ and $C$ are still given by \dref{2020918849n=2}.
 By virtue of the observer \dref{20209272056} and the   invertible transformation
 \begin{equation} \label{20209141545}
 UAU^{-1}=  A^{\top}, \  UB=C^{\top}, \ CU^{-1}=B^{\top},\ \ U=\begin{bmatrix}
   0&1 \\
   1&a_2
 \end{bmatrix}.
\end{equation}
the EDO of system \dref{20209272135}  becomes
\begin{equation} \label{20209272139}
  \left\{\begin{array}{l}\disp
\disp \dot{ \hat{z}}_1 (t) = \hat{z}_2(t)
- [(\alpha_2+\alpha_3)\omega_o-a_2][y(t)- \hat{z}_1(t)],\crr
 \disp \dot{ \hat{z}}_2 (t) = a_1\hat{z}_1(t)+a_2\hat{z}_2(t)+
  \alpha_1\alpha_3\omega_o^3\hat{v}(t)\\
- [(\alpha_1-\alpha_2\alpha_3)\omega_o^2 + a_2(\alpha_2+\alpha_3)\omega_o-a_2^2-a_1] [y(t)- \hat{z}_1(t)] +u(t),\crr
\disp \dot{ \hat{v}} (t) =[y(t)- \hat{z}_1(t)]
  ,
\end{array}\right.
 \end{equation}
where    $\alpha_1 $, $\alpha_2$ and $\alpha_3$ are constants such that  $\alpha_3<0$ and      the matrix
\dref{20209271122} is Hurwitz.
When  $a_1=a_2=0$, observer \dref{20209272139} is reduced to
\begin{equation} \label{20209272139a1a2}
  \left\{\begin{array}{l}\disp
\disp \dot{ \hat{z}}_1 (t) = \hat{z}_2(t)
-  (\alpha_2+\alpha_3)\omega_o [y(t)- \hat{z}_1(t)],\crr
 \disp \dot{ \hat{z}}_2 (t) =
  \alpha_1\alpha_3\omega_o^3\hat{v}(t)
-  (\alpha_1-\alpha_2\alpha_3)\omega_o^2   [y(t)- \hat{z}_1(t)] +u(t),\crr
\disp \dot{ \hat{v}} (t) =[y(t)- \hat{z}_1(t)]
 \end{array}\right.
 \end{equation}
and at the same time,
 system \dref{20209272135} turns to be
the canonical form of ESO in \cite{GAO2003} or high-gain observer in \cite{KhalilTAC2008}.
In this case, the  extended state   observer or high-gain observer   of system \dref{20209272135} is
\begin{equation} \label{20209272145}
  \left\{\begin{array}{l}\disp
\disp \dot{ \hat{z}}_1 (t) = \hat{z}_2(t)
- \beta_1\omega_o[y(t)- \hat{z}_1(t)],\crr
 \disp \dot{ \hat{z}}_2 (t) =  \ \hat{v}(t)
- \beta_2\omega_o^2 [y(t)- \hat{z}_1(t)] +u(t),\crr
\disp \dot{ \hat{v}} (t) =-\beta_3\omega_o^3[y(t)- \hat{z}_1(t)]
  ,
\end{array}\right.
 \end{equation}
where $\beta_1,\beta_2$ and $ \beta_3$ are constants such that the following matrix is Hurwitz
   \begin{equation} \label{20209272149}
 \begin{bmatrix}
 -\beta_1&1&0\\
 -\beta_2&0&1\\
 -\beta_2&0&0
\end{bmatrix} .
 \end{equation}
 By proper choices of $\beta_1,\beta_2$ and $\beta_3$, the observers \dref{20209272145} and \dref{20209272139a1a2} are
 equivalent  under  an invertible coordinate transformation.
 From this point,  the proposed EDO with constant dynamic $G=0$ covers the ESO as a special case and improves the ESO to the general observable linear system with input disturbance.

\section{Extended dynamic observer with harmonic dynamics}\label{Se.7}
This section devotes to a more general case than the constant dynamics discussed in Section \ref{ESO}.
 In most of engineering applications,
the disturbance is not completely ignorant.
 Some prior
 information about the disturbance usually has been known before the observer design.
When such a  prior
 information  is completely   known, i.e., the disturbance dynamics $G$ is known, the observer
can be designed    by
Theorem \ref{Th20201141911}.  When we only known a  roughly  prior information about the disturbance,
we then need both the high-gain and the known disturbance dynamics to deal  with the disturbance. 



To  make it more easier  to use, this section shows how to choose the dynamics of
disturbance by proper  choice of $G$.
By Remark \ref{Re20201222110},
 the  steady-state   error  of   the observer \dref{2020919191147}
  is   proportional to the error of the optimal  approximation $\|(I-\mathbb{P}_G)d\|_{\mathbb{S}}$  and is
 inversely proportional to    $\omega_{o}$.
 In order to decrease the steady-state   error  we should choose $G$ such that  $\Omega(G)$ is as large as possible.
 On the other hand,  the increment   of the order of $G$ may lead to overshoot in the
   transient response due to the high-gain and the extended order of disturbance dynamics. This,  in turn, makes us
    reduce the order of $G$   as much  as possible.
 Hence, we  need to find a trade-off between the observer accuracy and the response  performance.

Suppose that we have known that
$d\in \mathbb{S}$ is a continuous  periodic   signal
  with roughly known frequencies
  $\omega_j$, $j=1,2,\cdots,N$.
     In other words, the disturbance can be decomposed  into
   $d(t)=d_1(t)+d_2(t)$, where the non-constant dynamics of $d_1(\cdot) $ are completely unknown  and the dynamics of $d_2(\cdot)$ are   known., i.e.,
  \begin{equation}\label{2020261930}
d_2 (t)=   \sum\limits_{j=0}^{N}\left(a_j\cos\omega_jt+b_j\sin\omega_jt\right),
\end{equation}
 where $ a_j,b_j\in\R$,  $j=1,2,\cdots,N$   are unknown  amplitudes.
 By virtue of the  prior information  about the frequencies, we are able to
choose  $g_1,g_2,\cdots,g_{2N+1}$ such that
the matrix $G$ given by \dref{20209923729} with $m=2N+1$ satisfies
 $\sigma(G)=\{0,\pm \omega_j i\    |\ j=1,2,\cdots,N\}$.  Thanks to the Vieta theorem,  the choice of
  the parameters  $ g_{0 }, g_1,  g_{2}, \cdots,    g_{ 2N+1  }  $  is easy and  implementable.
Owing to  \dref{20195181211},  we have
 $d_2\in \Omega(G)$. By Theorem \ref{Th20209182115},       all the negative effects of $d_2(\cdot)$
 can be  eliminated
and the   steady-state error of   observer   \dref{2020919191147}  now
is   proportional to
 \begin{equation}\label{f20201282226hat}
  \|(I-\mathbb{P}_G)d \|_{\mathbb{S}}< \|\dot{d}_1 \|_{\infty}.
  \end{equation}
 If some frequencies  of $\omega_j$ are  large,
  then
$\|\dot{d}\|_{\infty}$ may be large as well. As the  result,
the   ESO or high-gain observer  may be  invalid  since the  observer
gain can not be arbitrarily large in engineering application.
  However, the EDO can still work  well because the high frequencies disturbance  has been removed completely
   by the extended dynamics.

   The main advantage of this approach lies in that we only need   rough prior  information  about the disturbance. All the unknown parts or the wrong   prior  information  can be treated automatically by the high-gain. In this way, we can make use of the prior  information  as much as possible and at the same time, the strong robustness to  the disturbance is possessed by the new proposed EDO.

To make this new methodology more understandable,
   we give another  example to show   how to utilize the prior periodic  information
of  the  disturbance. Suppose that $d\in \mathbb{S}$ is a periodic disturbance with known period $T$.
  By Fourier expansion,
  \begin{equation}\label{20201291011}
d(t)
 =  \sum\limits_{j=0}^{N} a_j\cos \frac{ j\pi t}{T} +\sum\limits_{j=N+1}^{\infty} a_j\cos\frac{ j\pi t}{T}:=d_1(t)+d_2(t),
  \end{equation}
  where $a_j$,  $j=0,1,\cdots$, are  the Fourier coefficients.
 Since $\dot{d}\in L^{\infty}[0,\infty)$,   we have
  \begin{equation}\label{20201301013}
\dot{d}(t)
 =  \sum\limits_{j=1}^{\infty} \tilde{a}_j\cos \frac{ j\pi t}{T},\ \ \ \  \tilde{a}_j=a_j
 \frac{ j\pi t}{T},
 \ \ j=0,1,\cdots,
  \end{equation}
  which implies that
  $\tilde{a}_j\to0$ as $j\to\infty$.
  Hence, we can choose $N$ large enough such that the remainder $\|\dot{d}_2\|_{\infty}$ is  sufficiently small.
   By   Vieta's  theorem,
  we can choose $g_{0 }, g_1,  g_{2}, \cdots,    g_{ 2N  }$   such that $G$ given by   \dref{20209923729} satisfies $ \sigma(G)=\left\{ \pm\frac{ j\pi t}{T}i \ |\
    j=0,1,2,\cdots,N\right\}$.
 As a result, we have $d_1\in \Omega(G)$ and hence
 the    steady-state error of   observer    \dref{2020919191147}
 is   proportional to  $\|\dot{d}_2 \|_{\infty}  $  that may be  much smaller than
 $\|\dot{d}  \|_{\infty} $.
  In this way, we have improved the accuracy of the observer without using the high-gain.
If we have known the best $N$-terms approximation of the Fourier expansion, $d_1(\cdot)$ in \dref{20201291011} can be replaced by its best $N$-terms approximation. In this case,
 we may obtain the higher accuracy of the  observer  \dref{2020919191147}  by a smaller order $N$.
  Due to nonlinear characteristics of the best  $N$-terms approximation
 \cite[Section 3.8]{Approximation2004},
   the observer is then actually a ``nonlinear observer" about the disturbance, although it is still
   a linear one to the control plant.

\begin{remark}\label{Re2020928802}
If we choose  $g_{0 }= g_1=   \cdots=   g_{ m  }=0$ in \dref{20209923729},
 then
 \begin{equation}\label{20209141747}
 \left\{ a_0t^m+a_1t^{m-1}+\cdots+a_{m-1}t+a_{m}\ |\ a_j\in \R, j=0,1,\cdots,m, \ t\in\R\right\}=\Omega(G).
 \end{equation}
 Therefore,  the  EDO  still   works for polynomial signals or polynomial
    piecewise    signal in some sense.
 Moreover, the exponential signals can still be treated by EDO.
  For example, if
the dynamics  $G$  satisfies   $\sigma(G )=\{ 0, {\lambda}\} $, $\lambda>0$ and   the algebraic multiplicity of the eigenvalue $ \lambda$  is ${n}_{\lambda}$, the signals of the type
$t^{n_{\lambda}}e^{\lambda t}$ belong to $\Omega(G)$.

\end{remark}

 \section{Feedback linearization  }\label{SectionFeedback}

In this section, we discuss  output feedback stabilization for  system \dref{2020151017}.
Without loss of  the generality, we suppose that
$A$, $B$ and $C$ are given by \dref{202010021047}
	 with $b_1=1$ and
  $b_2=b_3=\cdots=b_n=0$.
   In this case,
  system \dref{2020151017}  is  always observable for    $\mathbb{S}$ due to   Lemma \ref{Lm20201021046}.
 By Theorem \ref{Th202010021237},  for any $d\in \mathbb{{S}}$,
 system $\dot{x}(t)=Ax(t)+B[d(t)+u(t)]$ admits a feedback
 \begin{equation} \label{2020928841}
 u(t)=F_{\omega_c}x(t),\ \
F_{\omega_c}=[f_1(\omega_c)\  \ f_2(\omega_c)\ \ \cdots\ f_n(\omega_c)] ,\ \omega_c>0,
\end{equation}
 such that
  \begin{equation} \label{20201031324}
 \lim_{t\to\infty}\|x(t)\|_{\R^n}\leq \frac{M\|d\|_{\infty}}{\omega_c},\ \
 \end{equation}
 where $f_j\in C[0,\infty)$, $j=1,2,\cdots,n$ and  $M$ is a positive constant which  is independent of $\omega_c$.

  By Theorem \ref{Th20209182115},  $\hat{x}(\cdot)$ and $Q\hat{v}(\cdot)$ are   estimations of
   $x(\cdot)$ and $d(\cdot)$, respectively, where $\hat{x}(\cdot)$ and $\hat{v}(\cdot)$ come  from the observer \dref{2020919191147}.
 Similarly to \dref{20201261144},   the  output feedback  stabilizing control can be designed as
 \begin{equation} \label{20201122005}
   \begin{array}{ll}
\disp u(t)& =-Q \hat{v} (t)+F_{\omega_c}\hat{x} (t) ,
\end{array}
\end{equation}
where
   the first term is used to compensate for  the disturbance and the second term
   is  the stabilizer.
    In view of the observer \dref{2020919191147}, the feedback law \dref{20201122005} leads to
 the closed-loop system of \dref{2020151017}:
  \begin{equation} \label{20201122004}
 \left\{\begin{array}{l}
\disp \dot{x} (t) = A  x (t)+B  d(t)   -BQ \hat{v} (t)+BF_{\omega_c}\hat{x} (t) , \crr
\dot{ \hat{x}} (t) =
[A+( K_{\omega_o} +SE )C]\hat{x}(t)  -( K_{\omega_o} +SE )Cx(t)+BF_{\omega_c}\hat{x}(t),\crr
\disp \dot{ \hat{v}} (t) =G\hat{v}(t)-E C \hat{x}(t)+
 E Cx(t),
\end{array}\right.
\end{equation}
where $G$, $E$,   $K_{\omega_{o}},  S $ and $Q$ are chosen by the  scheme of observer \dref{2020919191147}.

 \begin{theorem}\label{th20201121836}
Under the Assumption \ref{Assum1},    for any $d\in  \mathbb{S} $,
there exist
$F_{\omega_c}$,  $K_{\omega_{o}},  S $ and $Q$  such that
 the solution of  closed-loop system   \dref{20201122004} satisfies:
  \begin{equation} \label{20201261548A}
   \disp \lim_{t\to\infty}\|x(t)\|_{\R^n}
 \leq
     \frac{M_0 \|(I-\mathbb{P}_G)d\|_{\mathbb{S}}}{\omega_{c}\omega_{ o}},\ \   \ \ \forall\ t\geq0,
\end{equation}
 where $\mathbb{P}_G$ is defined by \dref{20201211811},     $\omega_{o} $,  $\omega_{c} $  are tuning gains and $M_0$ is a positive constant that is independent of $\omega_{o} $.
Moreover,   $F_{\omega_c}\in  \R^{1\times n}$ can be chosen by Theorem \ref{Th202010021237} and
the observer parameters $G, E, K_{\omega_{o}},  S $ and $Q$ can be
  chosen by the  scheme  of parameters choice of observer \dref{2020919191147}.

     \end{theorem}
\begin{proof}
 Since  $d\in \mathbb{S}$, it can be represented dynamically as \dref{20201212010}.
By Theorem \ref{Th20209182115}, the observer \dref{2020919191147} is well-posed.          Define   the invertible transformation
   \begin{equation} \label{20201281203}
\begin{bmatrix}
x\\
v\\
\tilde{x} \\
\tilde{v}
\end{bmatrix}=
\begin{bmatrix}
I_n &0&0&0\\
0&I_{m+1} &0&0\\
I_n&0 &-I_n&0\\
0&I_{m+1} &0&-I_{m+1}\\
\end{bmatrix}\begin{bmatrix}
x  \\
v\\
\hat{x} \\
\hat{v}
\end{bmatrix}.
 \end{equation}
 In view of   \dref{20201212010}, the transformation \dref{20201281203}  converts the closed-loop system \dref{20201122004} into
    \begin{equation} \label{20201281209}
 \left\{\begin{array}{l}
 \disp   \dot{x} (t)= (A+BF_{\omega_c})  x (t)+B F_{\omega_c} \tilde{x} (t) +BQ  \tilde{v} (t), \crr
 \disp    \dot{ v} (t) =  Gv  (t)+\frac{B_d}{  Q B_d } \dot{e}(t),\crr
 \disp \dot{\tilde{x}} (t) = [A +( K_{\omega_o} +SE)C ]\tilde{x} (t)+B Q \tilde{v} (t),\crr
\disp  \dot{\tilde{v}} (t) =  G \tilde{v} (t) -EC\tilde{x}(t)+  \frac{B_{d}}{QB_{d}}\dot{e}(t).
\end{array}\right.
\end{equation}
By Theorem \ref{Th20209182115},
      there exists a positive constant $M_1$, independent of $\omega_o$ and $\omega_c$, such that
\begin{equation} \label{202093834}
\disp  \lim_{t\to\infty} |  F_{\omega_c}\tilde{x} (t)+Q\tilde{v} (t)  |  \leq
\frac{M_1    \|{e}\|_{\mathbb{S}}  }{\omega_o }.
\end{equation}
By Theorem \ref{Th202010021237},  \dref{202093834} and
  \dref{20201031324},
   there exists an $M_2>0$ such that
  \begin{equation} \label{202093838}
\disp\|x(t)\|_{\R^n}  \leq  \frac{M_2M_1    \|{e}\|_{\mathbb{S}} }{\omega_c\omega_o},
\end{equation}
which leads to
\dref{20201261548A}.
 \end{proof}



%
%
%

 \begin{remark}\label{Re202012131243}
 When the input disturbance is the nonlinear dynamics of the control plant,
 the EDO based feedback \dref{20201122005} actually
  achieves the feedback linearization of nonlinear system.
  After canceling the the nonlinear dynamics by its estimation,
   the transient performance
of the nonlinear system  behaves  like the
  nominal linear system.

 \end{remark}

\begin{remark}\label{Re20201251814}
 The disturbance with unknown dynamics is dealt with essentially by high-gain.
It is therefore  necessary to consider the sensitiveness to the  random measurement  noise. However,
  the strict theoretical analysis
   is not an   easy task. Here we  only give a simple numerical analysis in   Section \ref{Simulation}.
   A rigorous mathematical analysis is left in our next future  works.
Moreover, the ``peaking phenomenon" caused by high-gain and extended   dynamics  may take place in
the transient response. This drawback should be sufficiently taken into consideration in the practice.
\end{remark}

\begin{remark}\label{Re20201071100}
When $G$, $K$ and $P$ is given, the only tuning parameter of the observer \dref{2020919191147} is
$\omega_o$. Similarly, the only tuning parameter of the feedback \dref{2020928841}  is $\omega_c$ provided $f_j$ is given $j=1,2,\cdots,n$.
Therefore, the tuning parameters of the   closed-loop system \dref{20201122004}
can boil  down to to  $\omega_o$ and $\omega_c$
  which
  are referred to as ``bandwidth"
 of the observer and   controller, respectively in ADRC \cite{GAO2003}.
\end{remark}

\section{Numerical simulations }\label{Simulation}

In order to validate the developed   fundamental principle  visually, we present  some  simulations for the closed-loop system \dref{20201122004}.
 The finite difference scheme is adopted in discretization. The numerical results are programmed
in Matlab. The time step is  taken as $0.0001$.
Suppose that  the control plant is known and is given by \dref{2020918849n=2} with $a_1=2$ and $a_2=1$. Let $(x_1(0),x_2(0))=(0,1)$ and let the initial state of  the observer be zero.
  The tuning parameters are chosen as  $\omega_o=10$ and $\omega_c=10$ and
  the disturbance is chosen as
 $ d(t)=  \sin \omega t  +10$ with $\omega=10$.
   In contrast  with the simulations in \cite{KhalilTAC2008}, the frequency of disturbance  here is  much larger    but the tuning gain $\omega_o$ is much smaller.

 We consider three cases:  a) The only prior  information  about the disturbance is   $d\in \mathbb{S}$;
 b)~ There is an  estimation $9.5$ for  $\omega=10$; c)~ The frequency $\omega=10$ is known.
 We choose the extended dynamics as  $G_1=0$, $\sigma (G_2)=\{0,\pm9.5 i\}$ and
 $\sigma (G_3)=\{0,\pm10i\}$, respectively.
 The state estimation, disturbance estimation and the controller with $G_1$ are
 are plotted in Figure \ref{F1}. The counterparts for $G_2$ and $G_3$ are plotted in Figures \ref{F2} and
   \ref{F3}, respectively. In order to look at  the sensitiveness    of the measurement noise,
   the state estimation, disturbance estimation and the controller with $G_2$ and
   corrupted measurement $y(t)=Cx(t)+0.01\xi(t)$ are
      plotted in Figure \ref{F4}, where    $\xi(t)$ is the  standard Gaussian noise generated by the Matlab program command ``randn".

   \begin{figure}[!htb]\centering
\subfigure[$x_1 $ and its estimation $\hat{x}_1$]
 {\includegraphics[width=0.31\textwidth]{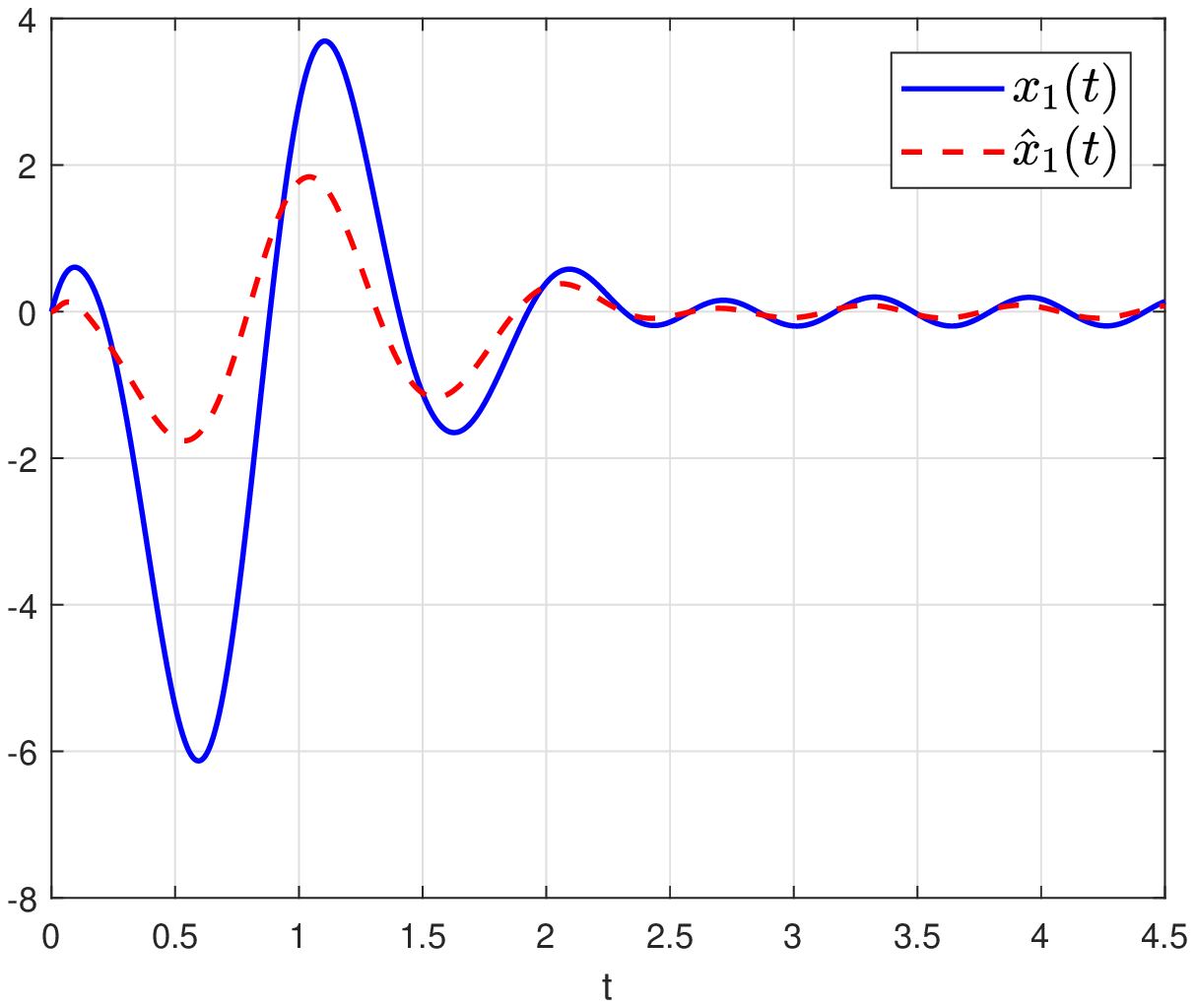}}
\subfigure[$x_2 $ and its estimation $\hat{x}_2$]
  {\includegraphics[width=0.31\textwidth]{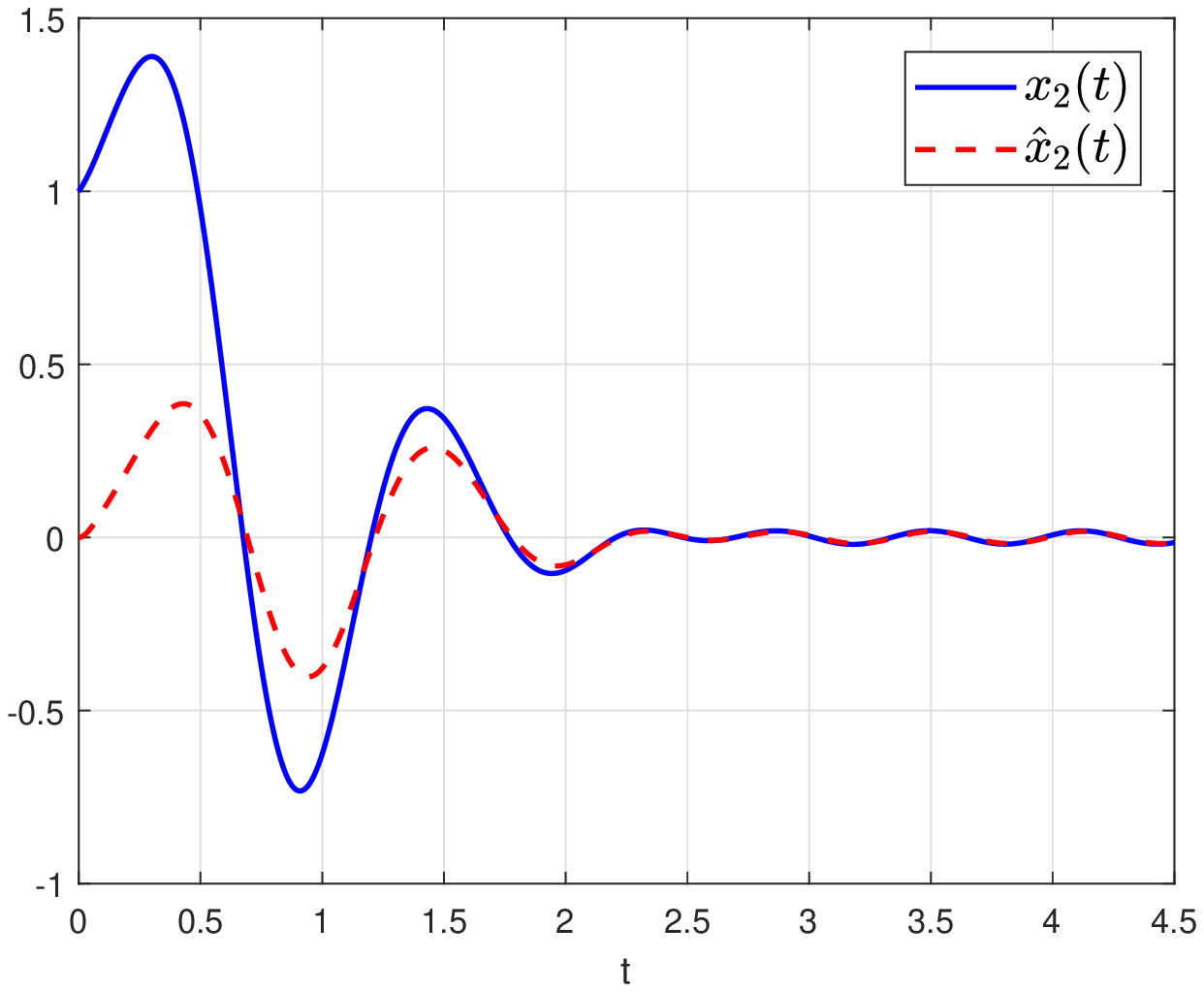}}
  \subfigure[$d -Q\hat{v}$ and  controller]
  {\includegraphics[width=0.31\textwidth]{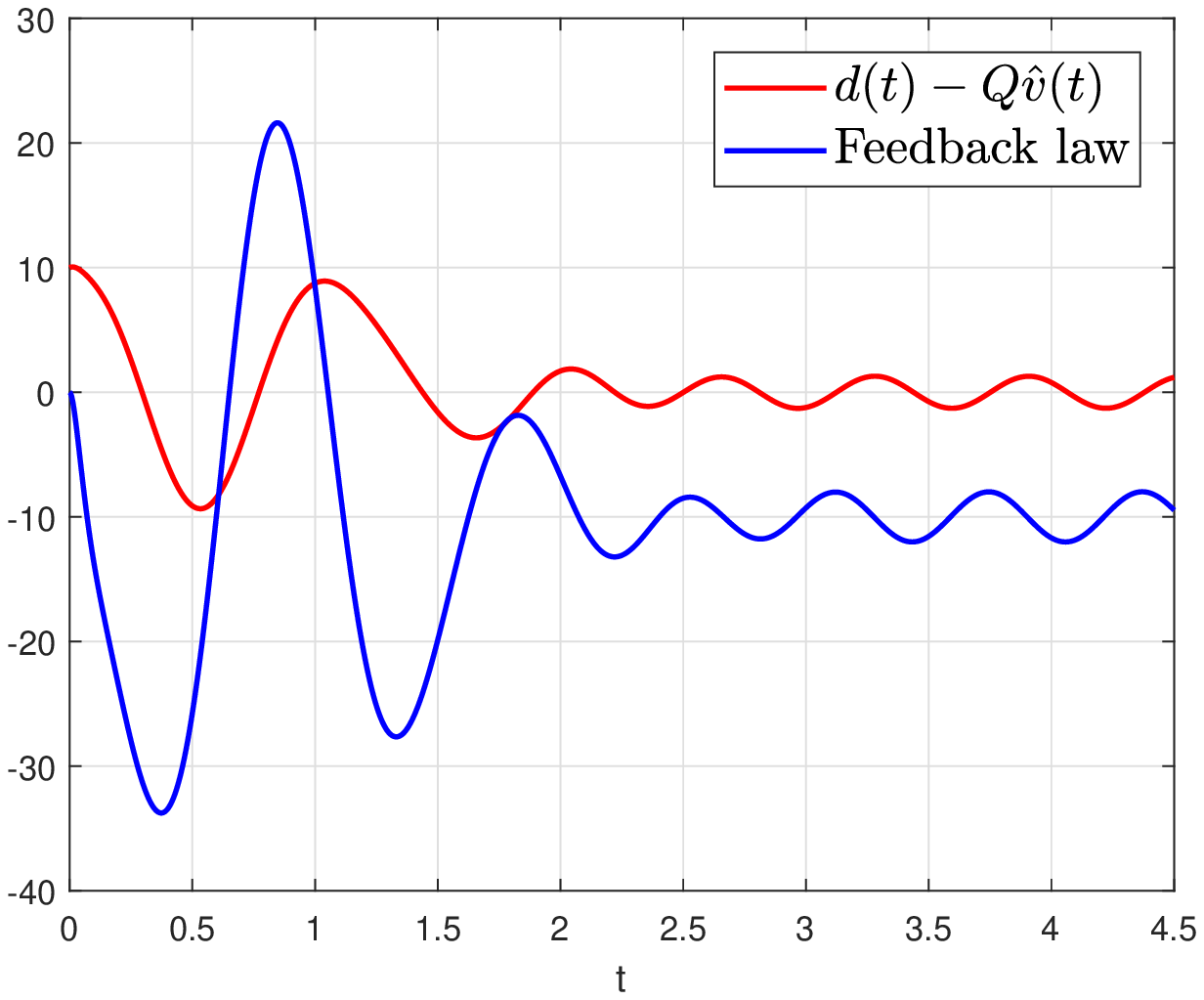}}
  \caption{The dynamics of disturbance is completely unknown except $d\in \mathbb{S}$.}
\label{F1}
\end{figure}

    \begin{figure}[!htb]\centering
\subfigure[$x_1 $ and its estimation $\hat{x}_1$]
 {\includegraphics[width=0.31\textwidth]{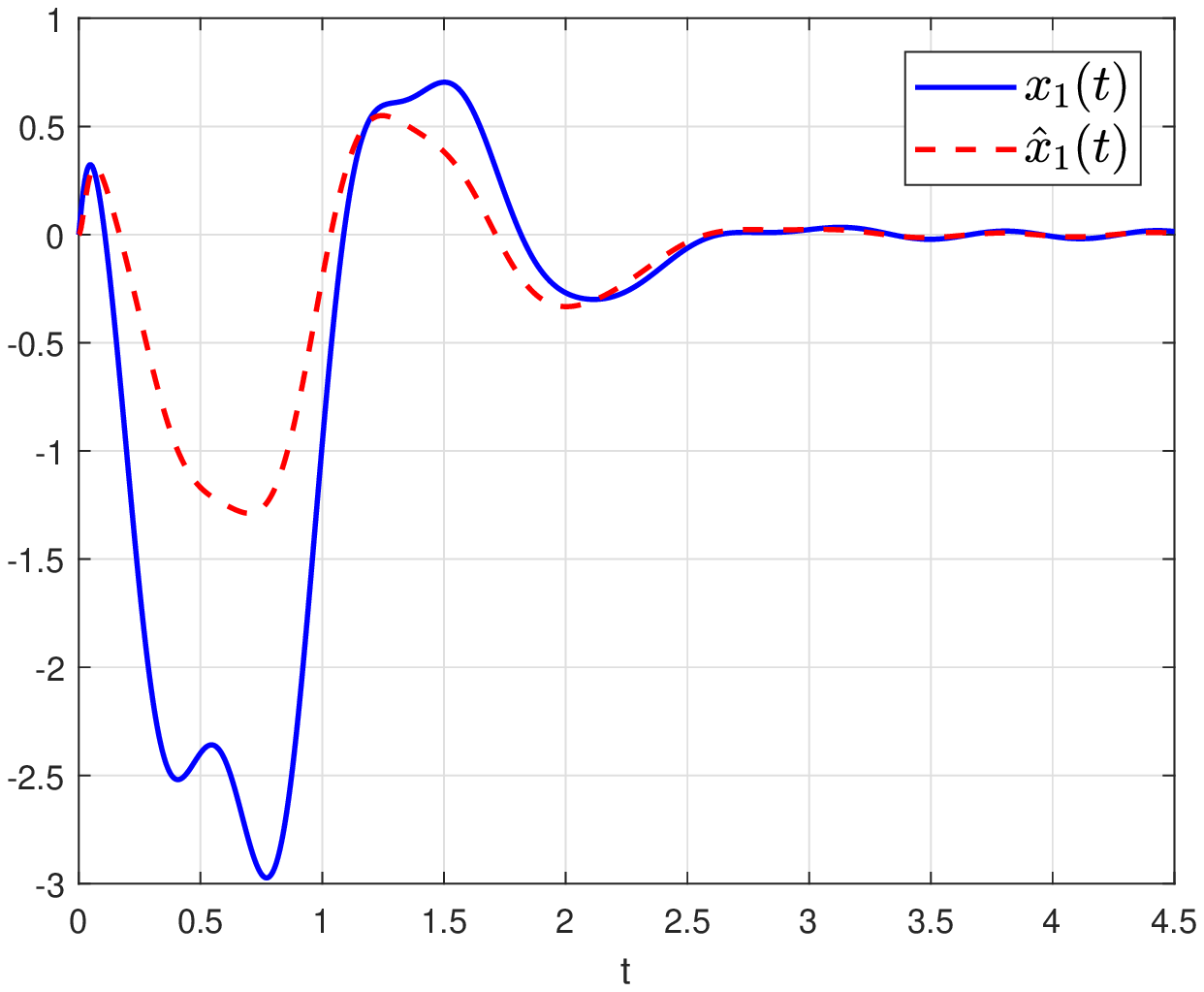}}
\subfigure[$x_2 $ and its estimation $\hat{x}_2$]
  {\includegraphics[width=0.31\textwidth]{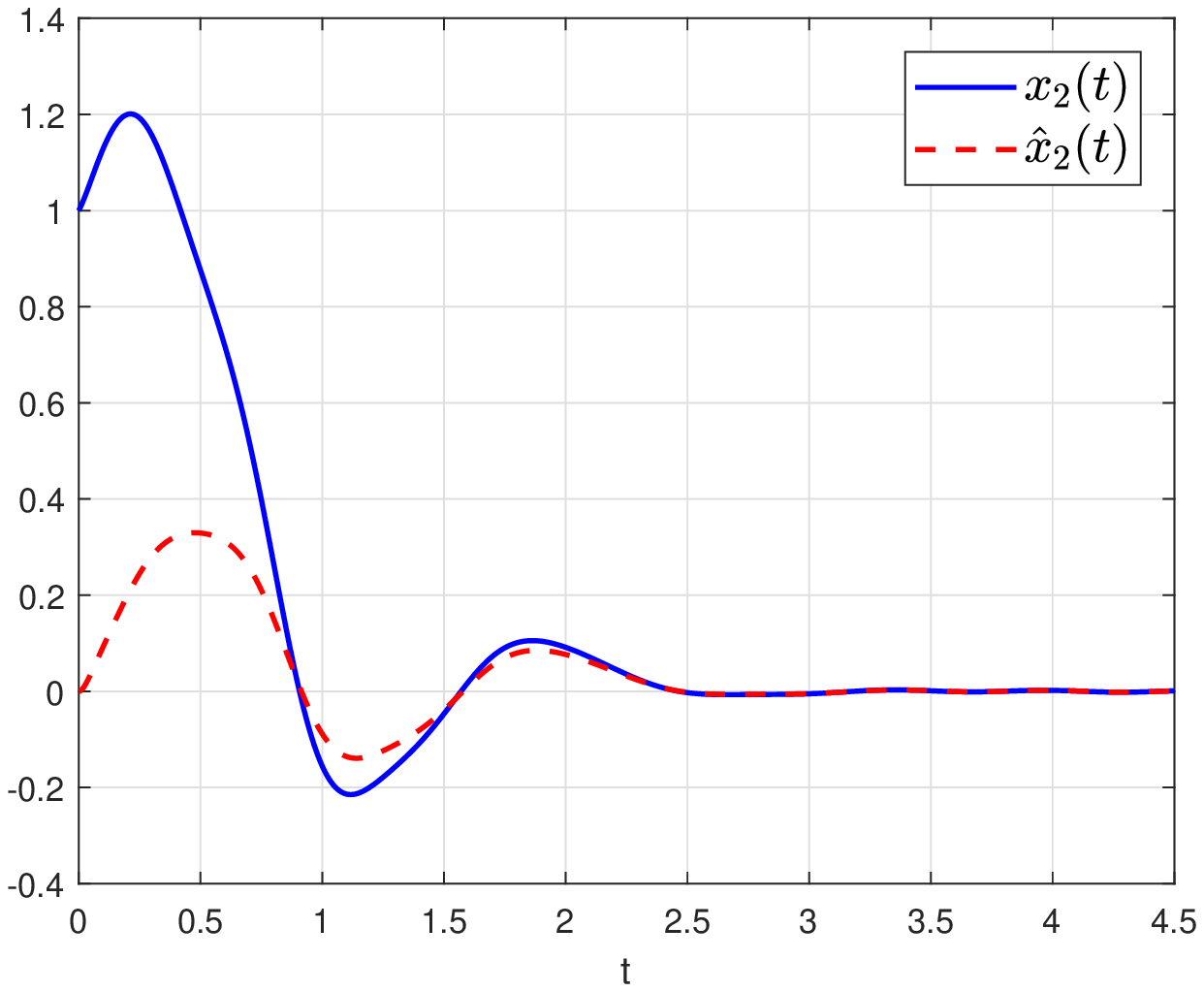}}
  \subfigure[$d -Q\hat{v}$ and  controller]
  {\includegraphics[width=0.31\textwidth]{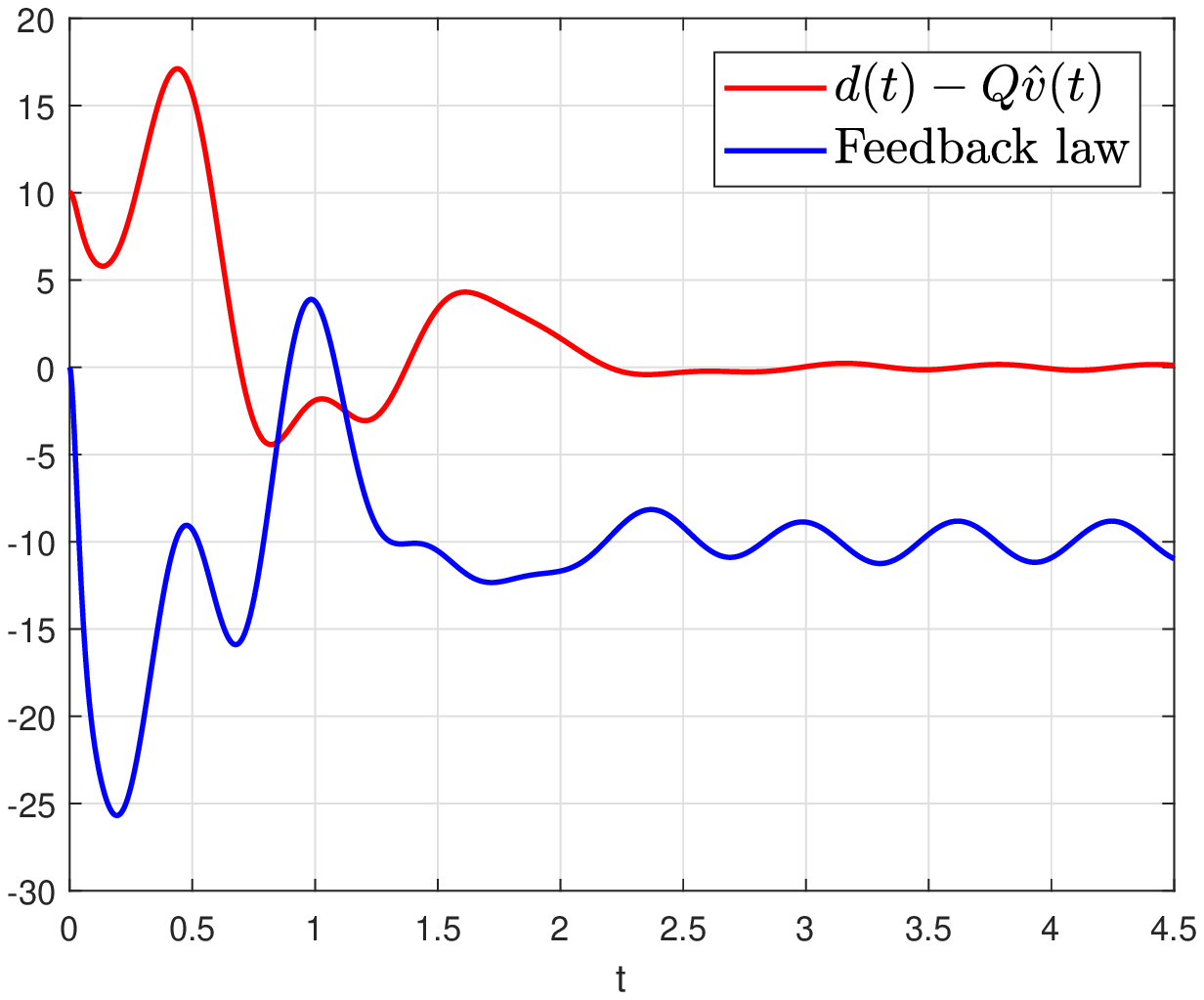}}
  \caption{The dynamics of disturbance is roughly known.}
\label{F2}
\end{figure}

    \begin{figure}[!htb]\centering
\subfigure[$x_1 $ and its estimation $\hat{x}_1$]
 {\includegraphics[width=0.31\textwidth]{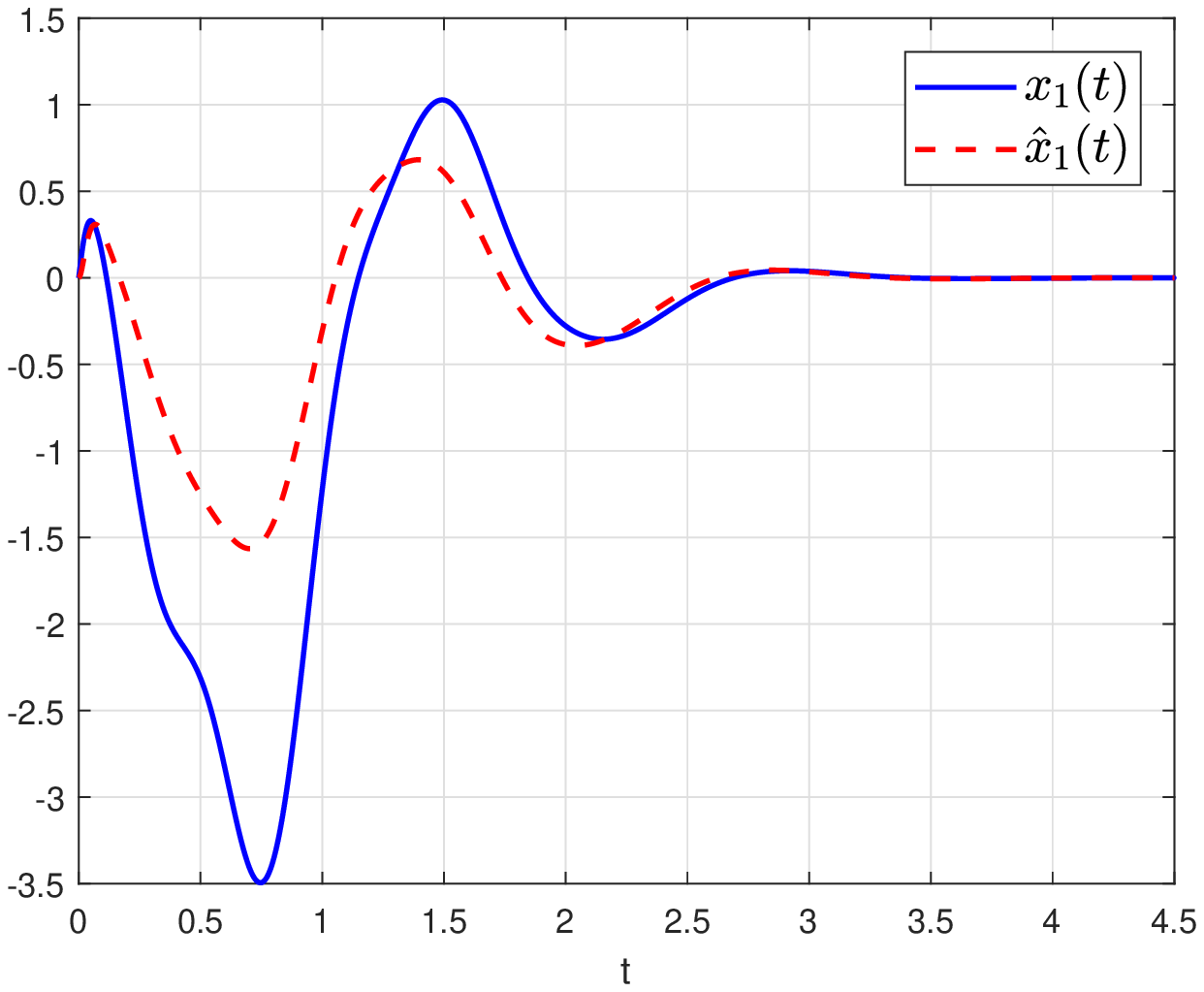}}
\subfigure[$x_2 $ and its estimation $\hat{x}_2$]
  {\includegraphics[width=0.31\textwidth]{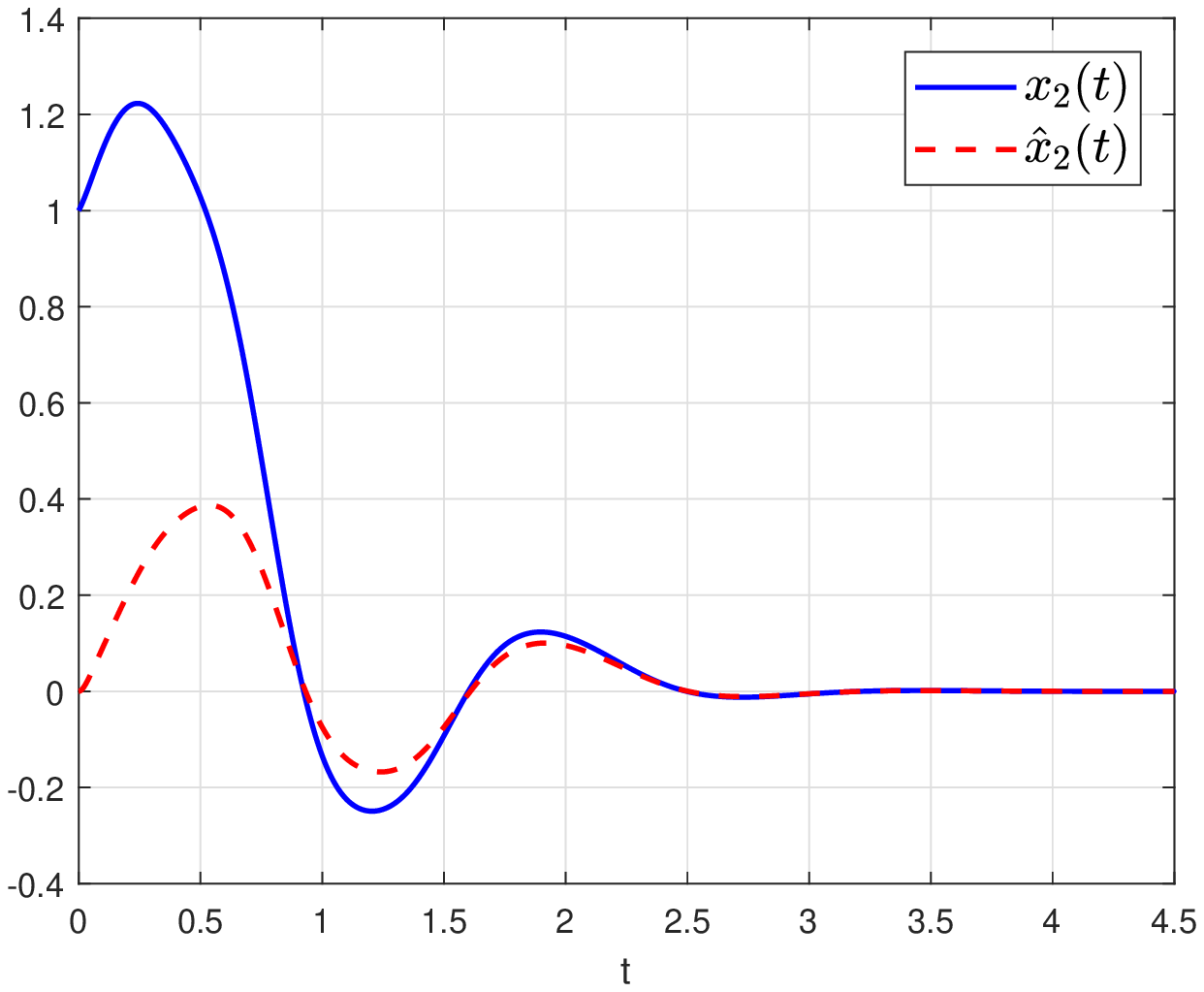}}
  \subfigure[$d -Q\hat{v}$ and  controller]
  {\includegraphics[width=0.31\textwidth]{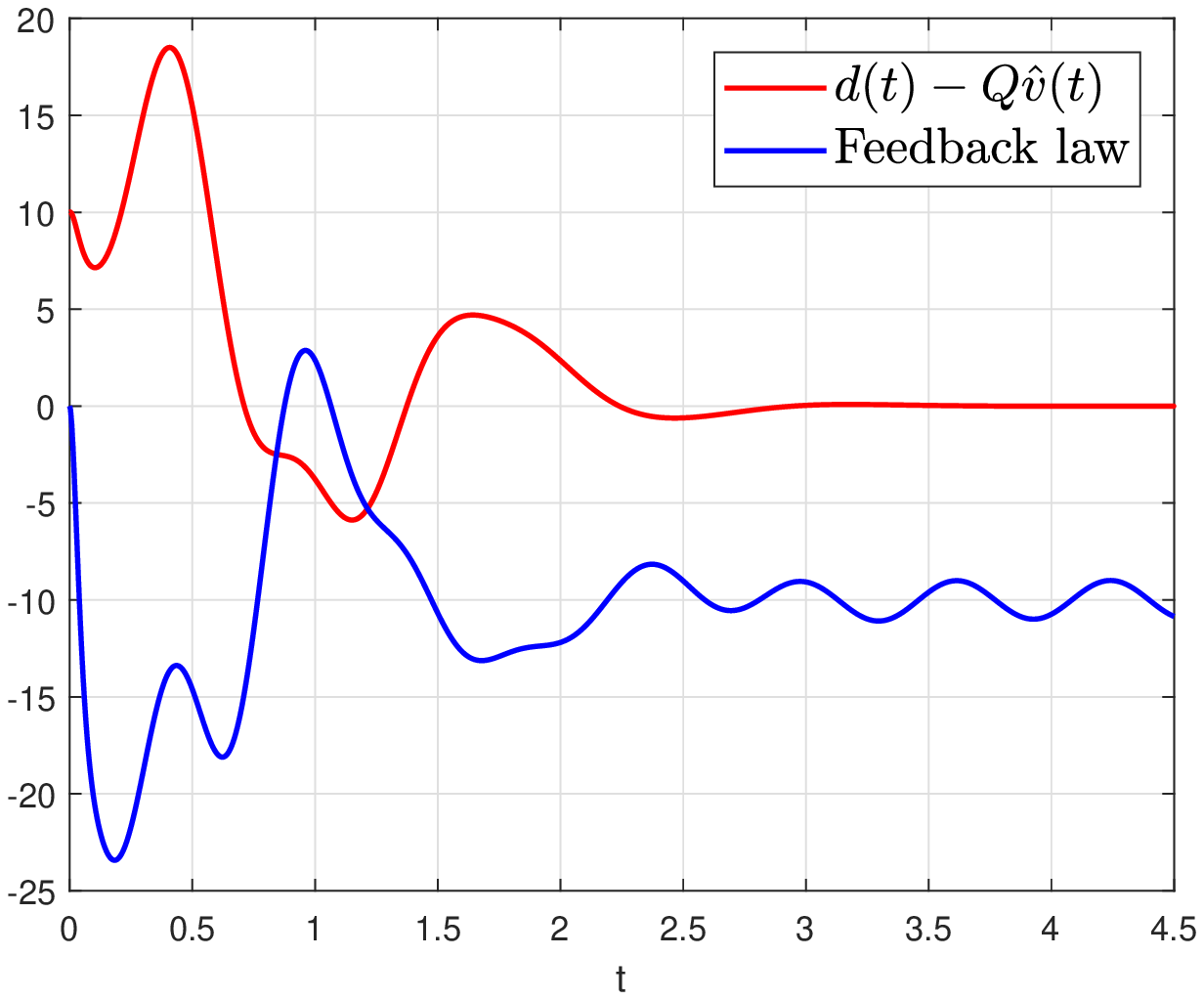}}
  \caption{The dynamics of disturbance is completely known.}
\label{F3}
\end{figure}

  \begin{figure}[!htb]\centering
\subfigure[$x_1 $ and its estimation $\hat{x}_1$]
 {\includegraphics[width=0.31\textwidth]{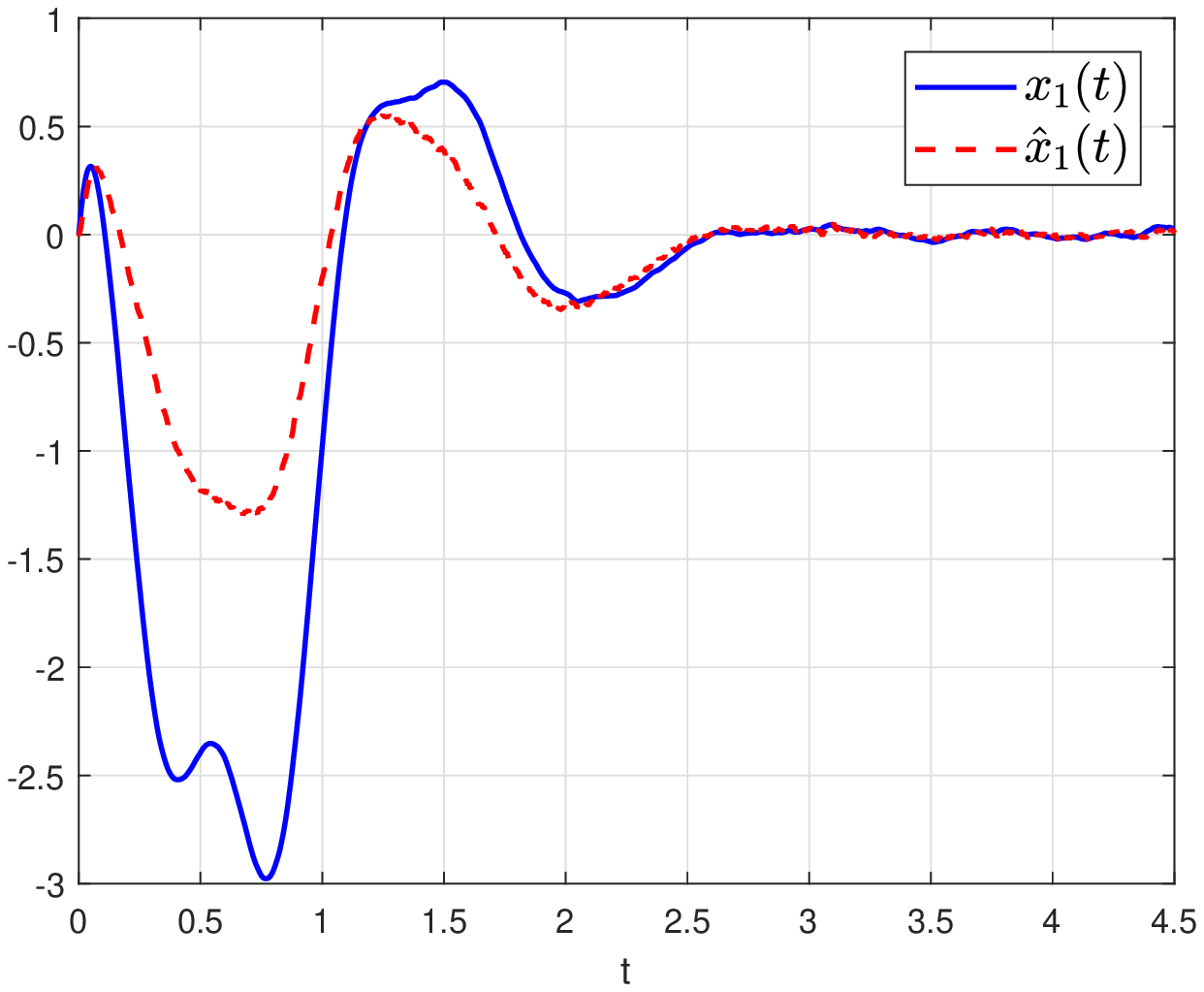}}
\subfigure[$x_2 $ and its estimation $\hat{x}_2$]
  {\includegraphics[width=0.31\textwidth]{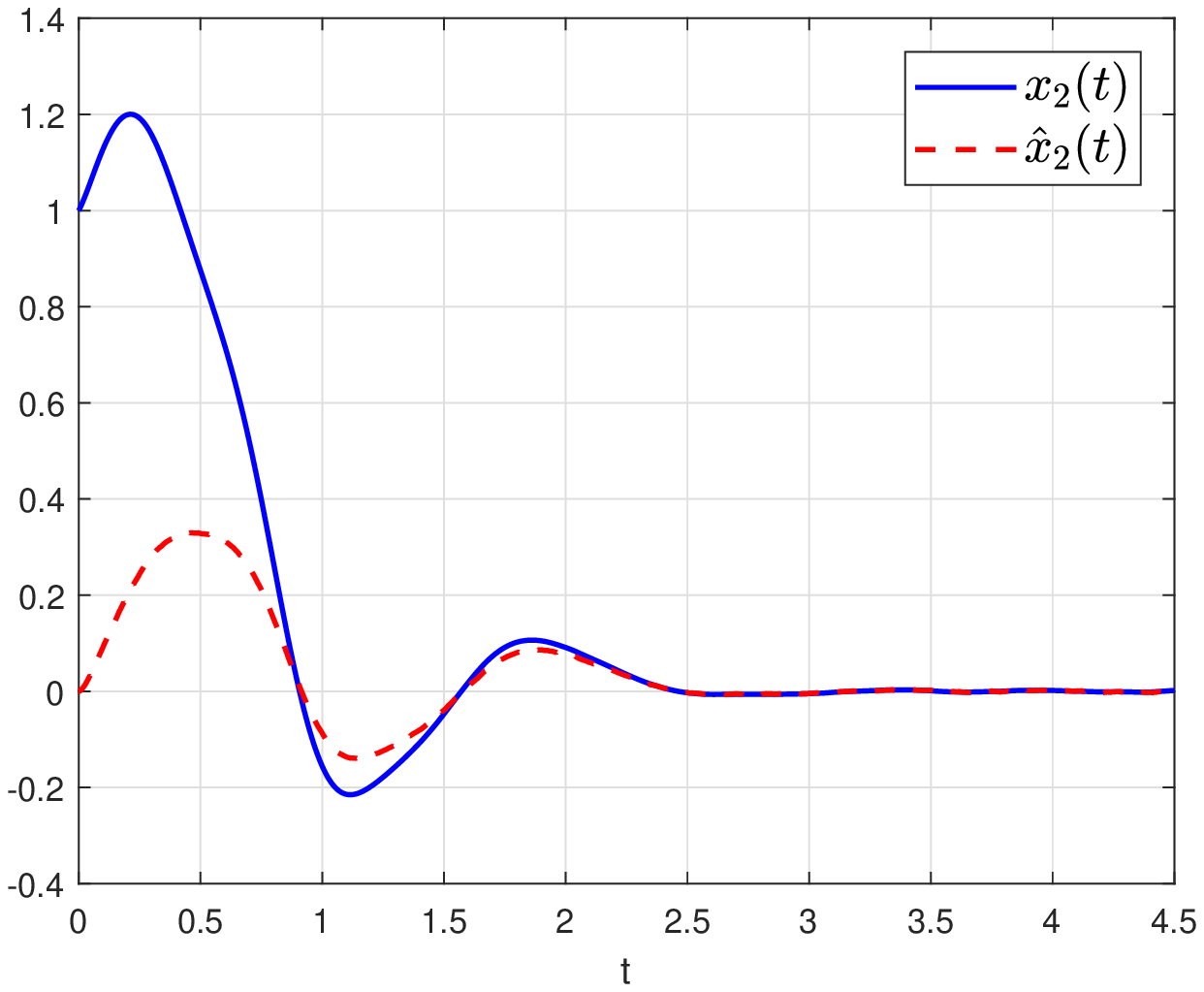}}
  \subfigure[$d -Q\hat{v}$ and  controller]
  {\includegraphics[width=0.31\textwidth]{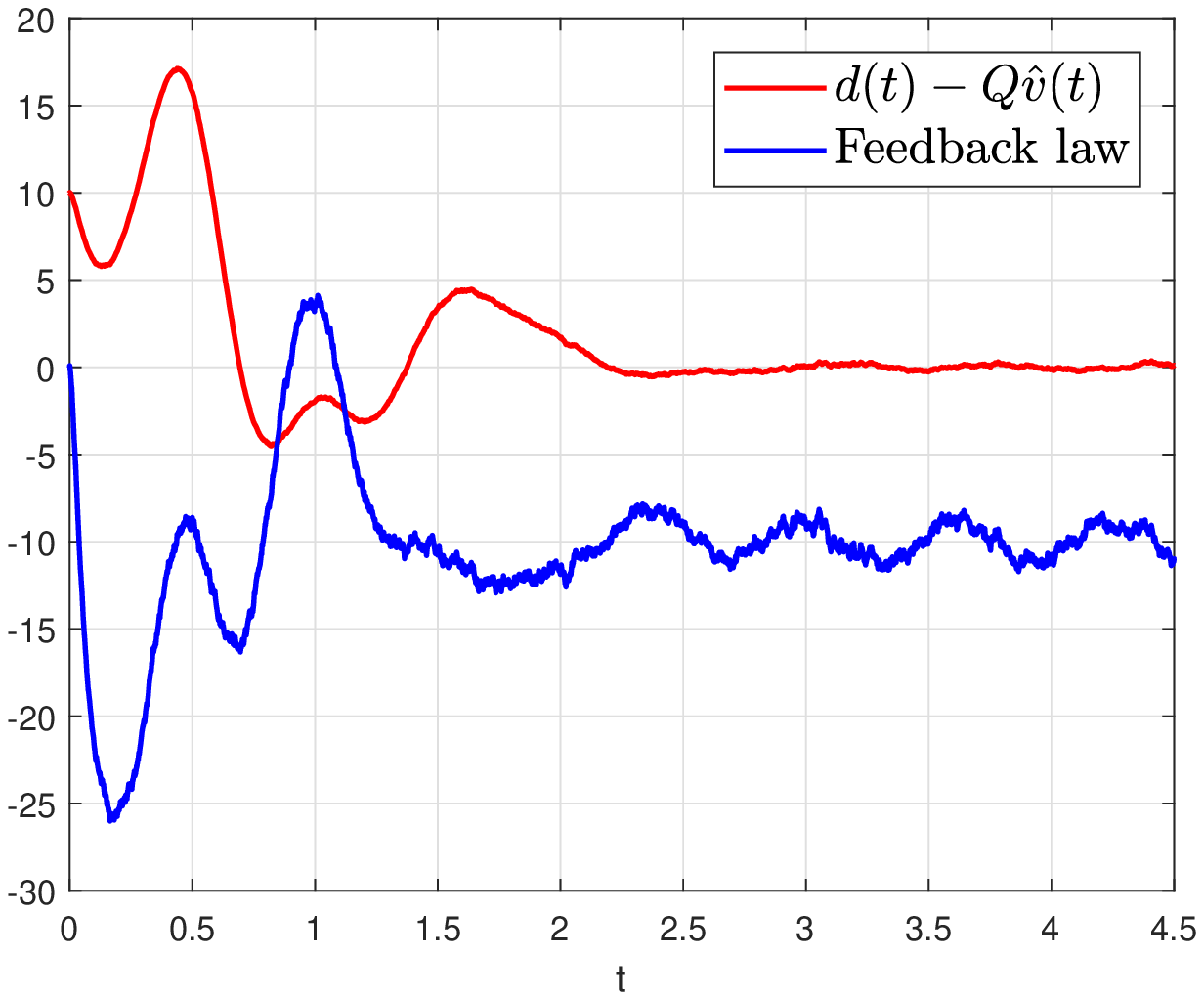}}
  \caption{The disturbance dynamics are $G_2$ and the measurement is corrupted by white noise.}
\label{F4}
\end{figure}

  Since the observer gain $\omega_o$  is relatively small, the error of the disturbance estimation
  is not very small for the case  $G_1=0$. However, if the have known the prior information
   $\sigma (G_2)=\{0,\pm9.5 i\}$, the accuracy of disturbance estimation is  improved significantly. When
   the disturbance dynamics are completely known, the  error of the disturbance estimation
   is convergent to   zero. Moreover, Figure \ref{F4} shows that the proposed EDO and its feedback are still insensitive to the measurement noise.


 Finally, we   point out that the peaking phenomenon takes place when  we improve more the convergent rate of the observer.
  This is caused by  the high-gain and the  order of  extended dynamics.
   In all  simulations, the output is technically  chosen as  $y(t)=(1-e^{- t})C[x_1(t)\ x_2(t)]^{\top}$ to avoid the peaking phenomenon.


\section{Conclusions}\label{conclusion}

In this paper, a novel  dynamics compensation approach is developed to stabilize   linear systems
 with input disturbance.  An extended dynamic observer (EDO)  is designed, in terms of both the prior information and the online measurement information,   to estimate both the disturbance and the system state simultaneously.
 The EDO takes almost all advantages from ESO  and IMP. More specifically, it
 possesses strong robustness to the system and disturbance, as the ESO in ADRC, and at the same time, it proposes a feasible way to utilize  as much  the prior  information  of the disturbance and the control plant
 as possible.
When there is no  information  about disturbance dynamics, the EDO is reduced automatically to an extension of ESO in ADRC    which has
  achieved great success
 in many engineering applications.



We just present   a fundamental principle for the observer and controller design. The technical tunings  such as shaping the transient response  are still required in engineering applications.
From the theoretical point of view,
  this paper gives a systematic way to   utilize the  prior disturbance information  and the high-gain. The future works are the online computations of the disturbance dynamics.




\end{document}